\documentclass[letterpaper, 11pt]{article}

\usepackage{CJKutf8}

\usepackage[whole]{bxcjkjatype}

\usepackage{bookmark}
\usepackage{type1cm}
\usepackage{enumerate}
\usepackage{algorithm}
\usepackage{algorithmicx}
\usepackage[noend]{algpseudocode}
\usepackage[sort]{natbib}

\usepackage{url}
\usepackage{hyperref}
\usepackage{threeparttable}
\usepackage[margin=1in]{geometry}
\usepackage{amsmath,amssymb,amsfonts,setspace} % math
\usepackage{amsthm} % math QED
\usepackage{latexsym}
\usepackage{comment}
\usepackage{color}
\usepackage{xcolor}
\usepackage{bm}
\usepackage[normalem]{ulem}
\usepackage{mathtools}
\usepackage{stmaryrd} % for double bracket

\usepackage[geometry]{ifsym}

\setcitestyle{numbers,square,comma}

\hypersetup{% hyperrefオプションリスト
setpagesize=false,
 bookmarksnumbered=true,%
 bookmarksopen=true,%
 colorlinks=true,%
 linkcolor=red,
 citecolor=blue,
}

\newtheoremstyle{mystyle}
    {}%                      % 上部スペース
    {}%                      % 下部スペース
    {\normalfont}%           % 本文フォント
    {}%                      % インデント量
    {\bf}%                   % 見出しフォント
    {}%                      % 見出し後の句読点, '.'
    { }%                     % 見出し後のスペース, ' ' or \newline
    {}%

\theoremstyle{mystyle}
\newtheorem{theorem}{Theorem}
\newtheorem{lemma}{Lemma}
\newtheorem{proposition}{Proposition}
\newtheorem{corollary}{Corollary}
\newtheorem{definition}{Definition}

\newcommand{\Clabel}{L^{\mathsf{con}}} %
\newcommand{\Cdec}{D^{\mathsf{con}}} %
\newcommand{\TClabel}{L^{\mathsf{tree}}} %
\newcommand{\TCdec}{D^{\mathsf{tree}}} %

\newcommand{\Olabel}{L^{\mathsf{out}}}%
\newcommand{\Odec}{D^{\mathsf{out}}} %
\newcommand{\hOlabel}{\hat{L}^{\mathsf{out}}}%
\newcommand{\hOdec}{\hat{D}^{\mathsf{out}}} %
\newcommand{\RSlabel}[1]{L^{\mathsf{RS}(#1)}}%
\newcommand{\RSdec}[1]{D^{\mathsf{RS}(#1)}} %

\newcommand{\Alabel}{L^{\mathsf{anc}}} %
\newcommand{\Adec}{D^{\mathsf{anc}}} %

\newcommand{\Cutset}{\partial}

\newcommand{\symdiff}{\triangle}

\newcommand{\Gcal}{\mathcal{G}}
\newcommand{\Scal}{\mathcal{S}}

\newcommand{\Ccal}{\mathcal{C}}

\newcommand{\Ecal}{\mathcal{E}}
\newcommand{\Qcal}{\mathcal{Q}}
\newcommand{\Xcal}{\mathcal{X}}
\newcommand{\Ycal}{\mathcal{Y}}
\newcommand{\Hcal}{\mathcal{H}}
\newcommand{\Zcal}{\mathcal{Z}}

\newcommand{\Hp}{\mathit{hs}}
\newcommand{\Et}{\mathit{ET}}
\newcommand{\NetFind}{\mathsf{NetFind}}

\newcommand{\Gext}{G^{\ast}}

\newcommand{\Allwords}{\ensuremath{\{0,1\}^{\ast}}}
\newcommand{\GFtwo}{\mathrm{GF}(2)}
\newcommand{\Field}{\mathbb{F}}

\newcommand{\Poly}{\mathrm{poly}}
\newcommand{\Polylog}{\mathrm{polylog}}

\newcounter{cntLemmaNumber}
\newenvironment{rlemma}[1]{%
\setcounter{cntLemmaNumber}{\thelemma}
\setcounterref{lemma}{#1}
\addtocounter{lemma}{-1}
\begin{lemma}
}{%
\end{lemma}
\setcounter{lemma}{\thecntLemmaNumber}
}
\newcounter{cntTheoremNumber}

\newcounter{cntPropositionNumber}

\allowdisplaybreaks[2]

\title{Deterministic Fault-Tolerant Connectivity Labeling Scheme}

\author{
Taisuke Izumi\footnote{Graduate School of Information Science and Technology, Osaka University, E-mail: t-izumi@ist.osaka-u.ac.jp.}
\and Yuval Emek\footnote{Tehcnion, Israel, E-mail: yemek@technion.ac.al.}
\and Tadashi Wadayama\footnote{Graduate School of Engineering, Nagoya Institute of Technology, E-mail: wadayama@nitech.ac.jp}
\and Toshimitsu Masuzawa\footnote{Graduate School of Information Science and Technology, Osaka University, E-mail: masuzawa@ist.osaka-u.ac.jp} 
}

\date{}

\begin{document}

\maketitle\thispagestyle{empty}
\addtocounter{page}{-1}

\begin{abstract}
The \emph{$f$-fault-tolerant connectivity labeling} ($f$-FTC labeling) is a scheme 
of assigning each vertex and edge with a small-size label such that one can determine the connectivity 
of two vertices $s$ and $t$ under the presence of at most $f$ faulty edges only from the labels 
of $s$, $t$, and the faulty edges. This paper presents a new deterministic $f$-FTC
labeling scheme attaining $O(f^2 \Polylog(n))$-bit label size and a polynomial construction time,
which settles the open problem left by Dory and Parter \cite{DP21}. The key ingredient of 
our construction is to develop a deterministic counterpart of the graph sketch technique 
by Ahn, Guha, and McGreger \cite{AGM12}, via some natural connection with the theory of error-correcting codes. This technique
removes one major obstacle in de-randomizing the Dory-Parter scheme. The whole scheme is 
obtained by combining this technique with a new deterministic graph sparsification algorithm derived from
the seminal \emph{$\epsilon$-net} theory, which 
is also of independent interest. As byproducts, our result deduces the first deterministic 
fault-tolerant approximate distance labeling scheme with a non-trivial performance guarantee 
and an improved deterministic fault-tolerant compact routing.
The authors believe that our new technique is potentially useful 
in the future exploration of more efficient FTC labeling schemes and other related 
applications based on graph sketches.   
\end{abstract}

\newpage

\section{Introduction}

\subsection{Motivation and Background} \label{sec:background}
Most message-passing distributed systems are modeled by graphs. 
By the nature of distributed computing, nodes in the network must cooperatively solve a given task 
without rich access to the whole topological information. In addition, the network is typically 
prone to faults, i.e., some of the vertices and/or links can be down by faults. Hence the distributed 
and compact representation of some property of the network (e.g., connectivity) adapting 
to topology modification is potentially useful for applications in distributed 
environments. The \emph{$f$-fault-tolerant connectivity labeling} ($f$-FTC labeling) is a scheme 
of assigning each vertex and edge with a small-size label. For any two vertices $s$ and $t$, and 
an edge set $F$ of $|F| \leq f$, it determines the connectivity of two vertices $s$ and $t$ under the deletion 
of edges $F$ only from the labels of $s$, $t$, and the edges in $F$. The concept of 
$f$-FTC labeling 
schemes (precisely, more general \emph{fault-tolerant distance labeling schemes} returning the $s$-$t$ 
distance rather than the $s$-$t$ connectivity) has been initiated explictly by Courcelle and Twigg~\cite{CT07}, 
following an earlier work by Feigenbaum, 
Karger, Mirrokni, and Sami~\cite{FKMS07}. The feature of FTC labeling schemes as a distributed data 
structure yields efficient structural algorithms for the \emph{forbidden set routing}
which routes packets avoiding a given set of faulty edges, and for more general \emph{fault-tolerant 
compact routing}~\cite{CT07,Chechik11,CLPR12,DP21,RDKR12} where the faulty edge set is initially unknown.

\subsection{Our Result}
While all of early results~\cite{ACG12,ACGP16,BCGMW21,CT07,CT10} mainly focus on the construction of 
small-sized labels for restricted graph classes, $f$-FTC labeling schemes for general graphs were 
recently proposed by Dory and Parter~\cite{DP21}. They propose two randomized $f$-FTC labeling schemes of 
$O(f + \log n)$-bit and $O(\log^3 n)$-bit label sizes respectively, which guarantee the weaker form 
of the correctness that the response to each \emph{single} query is correct with high probability. In other
words, they guarantee the correctness only for $1 - 1/O(\Poly(n))$ fraction of all possible $n^{O(f)}$ queries. 
We refer to this type of correctness criteria as ``whp query support'', in contrast with the standard criteria 
of ``full query support'' ensuring correct answers for all possible queries with high probability. The authors 
of \cite{DP21} also mention how the presented two schemes are converted to the ones with full query support, 
allowing the blow-up of their label sizes into $O(f\log n)$ bits and $O(f\log^3 n)$ bits respectively 
(see the footnote 4 of \cite{DP21}). 
In total, the paper \cite{DP21} presents the four randomized schemes, two of which attain 
full query support and the other two attain only whp query support. They leave as an open problem 
the polynomial-time \emph{deterministic} construction of compact FTC labeling schemes for general graphs. 
The main contribution of this paper is to settle this open problem:
\begin{theorem} \label{thm:maintheorem}
\sloppy{
There exist two deterministic $f$-FTC labeling schemes for any graph $G$ of $n$ vertices, $m$ edges,
and diameter $D$,
which respectively attain the following bounds:
\begin{itemize}
    \item The label size is $O(\log n)$ bits per vertex, and $O(f^2 (\log^2 n) \log\log n)$ bits 
    per edge. The query processing time is $\tilde{O}(|F|^4)$\footnote{The $\tilde{O}(\cdot)$ notation hides $\Polylog(n)$ factors.}, where $F$ is the set of queried edges 
    satisfying $|F| \leq f$. The construction time is polynomial of $m$.
    \item The label size is $O(\log n)$ bits per vertex, and $O(f^2 \log^3 n)$ bits 
    per edge. The query processing time is $\tilde{O}(|F|^4)$. The construction time is near linear, i.e., $\tilde{O}(mf^2)$. In addition, there exists a deterministic CONGEST distributed algorithm of constructing the labels in $\tilde{O}(\sqrt{m}D + f^2)$ rounds. 
\end{itemize}
%in $\tilde{O}(mf^2)$ time.
%The label size is $O(\log n)$ bits per vertex, and $O(f^2 \log^3 n)$ bits per edge. The query 
%processing time is $\tilde{O}(|F|^4)$
}
\end{theorem}
Note that every deterministic scheme inherently achieves full query support.
We emphasize that de-randomizing any of 
two original schemes of Dory and Parter is a highly non-trivial challenge. Those schemes are based 
on the other labeling schemes representing a sort of cut structures, whose construction heavily depends 
on randomness. 
Our deterministic construction is based on the second scheme of Dory and Parter~\cite{DP21} utilzing 
the \emph{graph sketch} technique by Ahn, Guha, and McGreger~\cite{AGM12} as the key structure. Informally, the graph sketch is 
a labeling scheme to edges,
admitting the detection of an outgoing edge for a given vertex set $S \subseteq V_G$ from 
the bitwise XOR of all labels of the edges incident to vertices in $S$. The technical highlight of 
our result is to develop a deterministic counterpart of the graph sketch technique 
via some natural connection with the theory of error-correcting codes. This technique is very simple, 
and 
completely removes one of two major obstacles in de-randomizing the outgoing edge detection 
by graph sketches. Yet another obstacle is the sparsification of the input graph. The sketch-based 
outgoing edge detection, including ours, works only when the input vertex set $S$ has 
a small number of outgoing edges. To handle the case with many outgoing edges, the original 
approach prepares a collection of spanning subgraphs, where for each possible input $S$ with 
a non-empty outgoing edge set, there exists at least one subgraph in the collection such that 
$S$ has exactly one outgoing edge. The construction of such a collection follows random 
sampling of edges. Our second contribution is a novel de-randomization technique for 
this graph sparsification process based on the seminal \emph{$\epsilon$-net} theory in computational 
geometry~\cite{HW87}. On this part, we present two different algorithms respectively corresponding 
to the schemes presented in Theorem~\ref{thm:maintheorem}. 

In addition to the key technical ideas above, the construction in Theorem~\ref{thm:maintheorem} 
are developed on the top of a few more notable features: First, 
our result is presented as a general framework with good modularity, and thus one can easily transform 
our deterministic scheme into an efficient randomized FTC labeling scheme with full query support, 
just by replacing the graph sparsification part with the conventional random edge sampling. The construction 
time and label size of this randomized scheme are competitive with the original sketch-based 
scheme in~\cite{DP21}.  
Second, we propose a new query optimization strategy. A drawback of our deterministic outgoing edge 
detection technique requires the decoding time roughly quadratic of the label size. Since the 
sketch-based $f$-FTC labeling scheme requires $|F|$ iterations of the outgoing edge detection for 
processing a single query, the straightforward implementations result in the $\tilde{O}(f^4|F|)$ time 
for the deterministic case, and $\tilde{O}(f^2|F|)$ time for the randomized case. Our query processing
algorithm shaves off this additional $|F|$ factor, as well as getting rid of the dependency on $f$ in
the outgoing edge detection. Consequently, we obtain a slightly improved randomized $f$-FTC labeling scheme 
of $\tilde{O}(|F|^2)$ decoding time. While the improvement of replacing $f$ by $|F|$ is very 
straightforward and easily applicable to any scheme not limited to ours, it is practically an intriguing 
feature because in typical scenarios the actual number of faults $|F|$ is substantially smaller than 
the upper bound $f$.

The detailed comparison
between the schemes in \cite{DP21} and our schemes are summarized in Table~\ref{tab:comparison}.

\subsection{Applications}

Our deterministic replacement of graph sketches provides a clearer insight to known outgoing edge 
detection techniques. It is simple, versatile, and potentially useful in the future exploration of other applications not limited to FTC-labeling schemes (e.g.,
~\cite{AGMR12,GK18,GP16,GKKT15,GP18,HPPSS15,JN17,KLMMS14,KW14,KKM13,KKT15,MK21}). 
Actually, we obtain several non-trivial de-randomization results in related topics.
It has been shown in \cite{DP21} that one can deduce the approximate distance labeling scheme, which provides an
approximate $s$-$t$ distance in $G - F$ given the labels of $s$, $t$ and the edges in $F$, utilizing any 
$f$-FTC labeling scheme in the black-box manner. In addition, such an approximate distance labeling scheme 
further deduces an efficient fault-tolerant compact routing scheme. Since the deduction parts are deterministically implemented, our deterministic $f$-FTC labeling schemes also de-randomize the construction of
the schemes above. More precisely, we obtain the following applications as corollaries of Theorem~\ref{thm:maintheorem}.
\begin{corollary} \label{corol:approxdistance}
Assume that the input graph is any weighted undirected graph with polynomially bounded edge weights.
For any positive intergers $k > 0$ and $f > 0$, there exists a $f$-fault tolerant $O(|F|k)$-approximate distance labeling scheme which achieves $\tilde{O}(f^2n^{1/k})$-bit label size and $\tilde{O}(|F|^4)$ query time.
\end{corollary}

\begin{corollary} \label{corol:compactrouting}
For any positive integer $k > 0$ and $f > 0$, there exist two deterministic fault-tolerant compact routing
schemes which achieve the stretch factor of $O(|F|^2k)$ and one of the following table-size bounds:
\begin{itemize}
    \item $\tilde{O}(f^2n^{1 + 1/k})$-bit total table size and $\tilde{O}(f^2n^{1 + 1/k})$-bit maximum local table size.
    \item $\tilde{O}(f^5n^{1/k})$-bit maximum local table size.
\end{itemize}
\end{corollary}
The result of Corollary~\ref{corol:approxdistance} is the first deterministic scheme for general graphs 
achieving a non-trivial performance guarantee. On Corollary~\ref{corol:compactrouting}, the prior work by 
Chechik \cite{Chechik11}, which attains $O(|F|^2(|F| + \log^2 n)k)$ stretch factor with a smaller table size, is 
also implemented deterministically, and thus our result is not the first deterministic solution. However, our scheme takes 
an advantage with respect to stretch factors. Since the proofs completely follow the reduction techniques proposed in \cite{DP21}, 
this paper does not present the precise formalism on these corollaries. See \cite{DP21} for details.

\begin{table*}[]
\caption{Comparison between the schemes in \cite{DP21} and our results. The dagger mark $\dag$ implies 
that the complexity is not explicitly stated in the original paper, and thus based on our analyses. 
For any scheme, the dependency on $f$ in query processing time is easily replaced by $|F|$ utilizing the
technique proposed in this paper.}
%While we do not present the comparison of the construction time, all the schemes attains a near linear construction time, i.e., $O(\Polylog(n))$ times of the output size.}
\label{tab:comparison}
\begin{center}
\begin{tabular}{rccccc}
            & label size        & query time          & Det./Rand. & correctness  & construction \\ \hline
1st (whp) \cite{DP21} & $O(f + \log n)$   & $\tilde{O}(f^3)$    & Rand. & whp     & $\tilde{O}(fm)$ \\
2nd (whp) \cite{DP21} & $O(\log^3 n)$     & $\tilde{O}(|F|)$      & Rand. & whp   & $\tilde{O}(fm)$ \\
1st (full)\cite{DP21} & $O(f \log n)$     & $\tilde{O}(f^3)^{\dag}$      & Rand.  & full & $\tilde{O}(fm)$    \\
2nd (full)\cite{DP21} & $O(f \log^3 n)$   & $\tilde{O}(f|F|)^{\dag}$     & Rand.  & full & $\tilde{O}(fm)$   \\
This paper  & $O(f^2 \log^3 n)$ & $\tilde{O}(|F|^4)$ & Det. & full  & $\tilde{O}(fm)$ \\
This paper  & $O(f^2 \log^2 n \log\log n)$ & $\tilde{O}(|F|^4)$ & Det. & full  & $\Poly(n)$ \\
This paper  & $O(f \log^3 n)$  & $\tilde{O}(|F|^2)$ & Rand. & full & $\tilde{O}(fm)$
\\ \hline
\end{tabular}
\end{center}
\end{table*}

\subsection{Related Work}

As mentioned in Section~\ref{sec:background}, FTC-labeling schemes, and more general
fault-tolerant (approximate) distance labeling schemes are introduced in the literature
of \emph{forbidden set routing}, which is the routing scheme avoiding non-adaptive faulty 
edges/vertices (i.e., the set of faults is not specified at the construction of routing tables, 
but given at the beginning of packet routing). The first result by Courcelle and Twigg~\cite{CT07} 
presents a FTC-labeling scheme of $O(k^2 \log n)$-bit labels for graphs of treewidth $k$, 
as well as its application to forbidden-set routing. A few 
results following this line exist~\cite{ACG12,ACGP16,BCGMW21,CT07,CT10}, but all of them are 
interested in the construction of compact labels for specific graph classes. The result 
for general graphs is not much addressed until the result by Dory and Parter~\cite{DP21}. 
In the context of deterministic construction, many of the results for restricted graphs stated above 
are deterministic, but the deterministic construction for general graphs is not known so far. 
In the paper of Dory and Parter, two randomized FTC labeling schemes relying on different 
techniques are presented. While the second scheme is based on graph sketches as we mentioned, 
the first one relies on the cut-verification labeling by Pritchard and 
Thurimella~\cite{PT11}.

There are many works in the literature of the centralized version of connectivity oracles and 
(approximate) distance oracles supporting edge/vertex deletion~\cite{BGLPSZ16,BCGLPP18,BK09,CLPR12,CCFK17,DT02,DP09,DP10,DP20,GRBMW21,GW19}. 
One of the major setting on this line is the case of $f = 1$, which is known as 
the \emph{replacement path problem}~\cite{GW19,GRBMW21,BK09}, or \emph{distance sensitivity oracle}~\cite{DT02,BK09,BGLPSZ16,BCGLPP18}. Roughly, the replacement path problem computes 
all pair (approximate) shortest path distances for every possible single edge/vertex fault. 
The sensitivity oracle is further required to store the information of replacement paths into 
a compact data structure. Their single-source variants are also investigated~\cite{BGLPSZ16,BK13}.
The sensitivity oracles for multiple faults are considered mainly in the context of 
\emph{fault-tolerant compact routing}, which is a generalization of forbidden-set routing addressing 
adaptive faults (i.e., the set of faults is not explicitly given at the beginning of packet 
routing)~\cite{Chechik11,CLPR12,DP21,RDKR12}. There are also a few results considering the 
oracles specific to connectivity~\cite{PT07,DP09,DP10,DP20}, which are seen as
centralized counterparts of FTC labeling schemes. Obviously, any $f$-FTC labeling scheme is also
usable as a centralized oracle with the space complexity of $m$ times of label size. More
general \emph{dynamic connectivity} of undirected graphs aims to develop the data structure 
of supporting the operations of inserting/deleting edges as well as the connectivity query. 
By definition, such a data structure can be used as a fault-tolerant connectivity oracle 
with the query processing time of $O(|F| \cdot \text{(operation cost)})$. While there is 
a long history of this problem~\cite{Thorup00,HLT01,HK99,CGLNPS20,Wulff-Nilsen13,PD06,KKM13,GKKT15}, 
all the known results with $\Polylog(n)$-time 
operation cost rely on amortized analyses or the correctness criteria of whp guarantee.
Focusing on deterministic construction, the best known results are the algorithm 
by Chuzhoy, Gao, Li, Nanongkai, Peng and Saranurak for dynamic connectivity achieving $n^{o(1)}$ operation 
cost per one edge deletion~\cite{CGLNPS20}, and the connectivity oracle by P\u{a}tra\c{s}cu and Throup \cite{PT07} 
achieving $\tilde{O}(|F|)$ query processing time. 
The deterministic algorithm for dynamic connectivity with worst-case $\Polylog(n)$-time operation 
is a major open problem in this research field. 
%It is an interesting research direction to investigate 
%if our new technique can provide any progress on on this open problem or not.

While this paper focuses only on edge faults, it is also an interesting research direction to 
consider vertex faults. Despite its similarity, vertex fault-tolerance often exhibits 
a technical difficulty quite different from edge fault-tolerance. A trivial approach is to reduce 
the failure of a vertex $v$ into the failure of all the edges incident to $v$, which results in 
a $f$-vertex fault tolerant connectivity labeling scheme of $\tilde{O}(\Delta f)$-bit label size 
(where $\Delta$ is the maximum degree of the input graph). Unfortunately, this approach does not 
provide a good worst-case bound because
$\Delta$ could become $\Omega(n)$. Recently, vertex fault tolerant labeling schemes for small $f$ 
are proposed by Parter and Petruschka~\cite{PP22}. They provide two schemes respectively attaining 
a polylogarithmic label size for $f = 2$ and a sublinear size for $f = o(\log\log n)$.

The graph sketch technique is first presented by Ahn, Guha, and McGreger~\cite{AGM12}, aiming to develop space-efficient algorithms 
for graph stream~\cite{AGMR12,KW14,KLMMS14}. There are a variety of applications not limited to graph stream, such as distributed computation~\cite{GK18,GP16,GP18,HPPSS15,JN17,KKT15,MK21} and dynamic algorithms~\cite{GKKT15,KKM13}.

\subsection{Roadmap}
\sloppy{
Following the introduction of necessary notations and definitions, we first explain the high-level idea of our framework in Section~\ref{sec:framework}, including the explanation of the sub-components constituting our scheme.
Section~\ref{sec:ourApproach} explains the key technical ideas of our de-randomization technique.
Following the summary of whole structure in Section~\ref{sec:wrapup}, we present a further 
query optimization technique in Section~\ref{sec:improvement}. Section~\ref{sec:details} explains the details of our construction.
We consider the distributed construction of our labeling schemes in Section~\ref{sec:distconst}. Finally, we conclude this paper in 
Section~\ref{sec:conclusion}, as well as a few promising future research directions.
}

%\section{Technical Outline} \label{sec:outline}

\section{Notations and Terminologies}
We denote the vertex set and edge set of a graph $G$ respectively by $V_G$ and $E_G$. 
We use the notation $H \subseteq G$ to represent that $H$ is a subgraph of $G$. 
For any edge subset $E' \subseteq E_G$, 
we define $G - E'$ as the graph obtained from $G$ by removing all the edges in $E'$. 
Given a vertex subset $S \subseteq V_G$, an \emph{outgoing edge} of $S$ is the edge 
having exactly one endpoint in $S$. We define $\Cutset_G(S)$ as the set of 
all outgoing edges of $S$ in $G$. For any $E' \subseteq E_G$, we also define $\Cutset_{E'}(S) = 
\Cutset_G(S) \cap E'$. 

Given a rooted tree $T$ and vertex $v \in V_T$, we denote by $T(v)$ the subtree 
of $T$ rooted by $v$. Given an edge $e \in E_T$, we also denote by $T(e)$ the subtree of $T$ rooted 
by the lower vertex (i.e., the endpoint farther from the root than the other) of $e$.

An \emph{$f$-FTC labeling scheme} for a given input graph $G$ consists of a labeling function $\Clabel_{G, f}$
and a universal decoding function $\Cdec_{f}$. The labeling function assigns each of vertices and edges 
$x \in V_G \cup E_G$ with a label $\Clabel_{G, f}(x)$ (i.e., a binary string). 
Let $s, t \in V_G$ be any two vertices and $F \subseteq E$ be any edge subset of size at most $f$.
The decoder function $\Cdec_f$ correctly answers the connectivity between $s$ and $t$ in $G - F$ 
only with the information of $\Clabel_{G, f}(s)$, $\Clabel_{G, f}(t)$, and $\{\Clabel_{G, f}(e) \mid e \in F\}$. 
Note that the decoder function $\Cdec_f$ is universal for all $G$, and cannot have any direct access to the 
information of $G$. The detailed formalism of $f$-FTC labeling scheme is given in Section~\ref{subsec:specification}.

\section{Construction Framework}
\label{sec:framework}
We first introduce a general framework of constructing the $f$-FTC labeling scheme. 
The technical core of this framework relies on the scheme by Dory and Parter~\cite{DP21}, but 
some additional techniques and abstractions are newly introduced. 
Let $G$ be the undirected input graph of $n$ vertices and $m$ edges. Throughout this paper, 
we fix an arbitrary rooted spanning tree $T$ of $G$. The framework 
consists of two technical components. The first one is a weaker variant of $f$-FTC labeling schemes 
which supports only the query $(s, t, F)$ satisfying $F \subseteq E_T$ 
(i.e., only the edges in $T$ can be faulty). We refer to this scheme as 
the \emph{tree edge $f$-FTC labeling scheme}. Our tree edge $f$-FTC labeling scheme is implemented 
with two other labeling schemes, respectively referred to as the \emph{ancestry labeling scheme} and 
the \emph{$\Scal$-outdetect labeling scheme} (explained in the next section). 
The second component is the very simple transformer which deduces an $f$-FTC labeling scheme from any tree edge $f$-FTC 
labeling scheme with no blow up of label size. In the following sections we explain the outline of each component.

\subsection{Tree Edge $f$-FTC Labeling Scheme}
\label{sec:treeEdgeFTC}

First, we state the informal specifications of the two sub-schemes.
The formal definitions of these sub-schemes are also presented in Section~\ref{subsec:specification}. 
\begin{itemize}
\item \textbf{Ancestry Labeling Scheme}:
Let $T$ be any tree. This scheme assigns 
each vertex $v \in V_T$ with a label $\Alabel_{T}(v)$. 
Given two labels $\Alabel_T(u)$ and $\Alabel_T(v)$ of distinct vertices 
$u, v \in V_T$, one can determine if $u$ is the ancestor of $v$, the descendant of $v$, or otherwise. 
There exists a linear-time deterministic algorithm which provides the ancestry labeling of $O(\log n)$ bits~\cite{KNR92}.
\item \textbf{$\Scal$-Outdetect Labeling Scheme}:
We assume that each edge $e \in E_G$ is assigned with a unique ID from some domain $\Ecal$.
Let $\Scal \subseteq 2^{V_G}$ be a collection of vertex subsets. An $\Scal$-outdetect labeling scheme 
assigns each vertex $v \in V_G$ with a label $\Olabel_{G}(v)$. For any vertex subset $S \in \Scal$ 
such that $\Cutset_G(S)$ is nonempty, one can compute an outgoing edge $e \in \Cutset_G(S)$ only 
from the bitwise XOR sum $\bigoplus_{v \in S} \Olabel_G(v)$ of all the labels assigned to vertices in $S$.
If $\Cutset_G(S)$ is empty, the scheme also detects it.
The \emph{graph sketch} technique~\cite{ACG12} provides a randomized $\Scal$-outdetect labeling scheme 
with $O((\log |\Scal|) \cdot \mathrm{polylog}(n))$-bit label size.
\end{itemize}
We define $\Scal_{f, T}$ as the collection of all vertex subsets $S$ satisfying $\Cutset_T(S) \leq f$. Note that $S$ is not required to induce a connected subtree of $T$. 
Roughly, our tree edge $f$-FTC labeling scheme is the combination of the ancestry labeling 
scheme of $T$ and the $\Scal_{f, T}$-outdetect labeling scheme of $G - E_T$ for an appropriate 
edge ID domain $\Ecal$ (explained later). Each vertex $u$ is assigned with the ancestry label 
of $u$, and each tree edge 
$e = (u, v)$ in $E_T$ is assigned with the concatenation of the ancestry labels of $u$ and $v$, and 
the XOR sum of the $\Scal_{f, T}$-outdetect labels over all the descendant vertices of $e$. 
We do not have to assign any label to non-tree edges because we focus on the construction of 
the tree edge $f$-FTC labeling scheme.

Given a query $(s, t, F)$ satisfying $F \subseteq E_T$ and $|F| \leq f$, the spanning tree $T$ is split into $|F| + 1$ subtrees 
by removing all the edges in $F$. We refer to the vertex set of each split subtree as a \emph{fragment}. 
Let $S$ be the fragment of containing $s$. The query processing algorithm iteratively grows $S$ by detecting 
an outgoing edge $e \in \Cutset_{G - E_T}(S)$. If such an edge is found, the fragment that $e$ reaches from $S$ 
is merged into $S$. This process terminates until no outgoing edge is found or the fragment with $t$ is merged. 
If no outgoing edge is found, one can conclude that $s$ and $t$ are not connected in $G - F$, or connected 
otherwise. Our framework detects an outgoing edge of $S$ in $G - E_T$ by 
the $\Scal_{f, T}$-outdetect labeling scheme (recall $S \in \Scal_{f, t}$ by definition). With 
the support of the ancestry labeling scheme, one can detect the ancestor-descendant relationship between 
any entities in $s$, $t$, and $F$, which provides the information of the edge set $\Cutset_{T}(S')$ for 
all fragments $S'$. To compute $\bigoplus_{v \in S'} \Olabel_{G - E_T}(v)$ for each fragment $S'$, it suffices to compute the XOR sum of the $\Scal_{f, T}$-outdetect labels assigned to the edges in $\Cutset_{T}(S')$. Since the outdetect label of each edge in $T$ is the XOR sum of all descendants' outdetect labels, it appropriately cancels out the labels assigned with 
the vertices not in $S'$.
The fragment merging is simple but has a point to be careful. Let $e = (u, v)$ be an outgoing edge 
of $S$ (assuming $u \in S$), and $S'$ be the fragment containing $v$. Then the XOR sum over all the 
labels of $S \cup S'$ is easily calculated by the XOR sum of the two computed sums for $S$ and $S'$. 
However, how can we identify the fragment $S'$ containing $v$? This problem is resolved by 
embedding the ancestry labels of $u$ and $v$ into the edge ID of $e$. 
That is, as a preprocessing step, we assign each edge $(u, v)$ with the pair of $(\Alabel_{T}(u), \Alabel_{T}(v))$ as the edge ID, and the outdetect labeling is constructed 
for the edge domain by this assignment. Then the decoding of an edge ID immediately yields the information of the fragments containing its endpoints. We formalize our framework explained above into the following lemma. 
\begin{lemma}\label{lma:framework}
Assume any deterministic $\Scal_{f, T}$-outdetect labeling scheme $(\Olabel_G, \Odec)$ of label size $\alpha$ and decoding time $\beta$. Then there exists a deterministic tree edge 
$f$-FTC labeling scheme of $(\alpha + O(\log n))$-bit label size and 
$O(|F|(\beta + \log |F|))$ decoding time, where $F$ is a set of faulty edges given by a query.
\end{lemma}
The proof details are given in Section~\ref{subsec:treeFTC}.

\subsection{Transformation to General Scheme} \label{subsec:transformation}
The second component reduces the construction of general $f$-FTC labeling schemes into 
that of tree edge $f$-FTC labeling schemes.
For the input graph $G$ and its rooted spanning tree $T$,
our transformation constructs an auxiliary graph $G'$ by subdividing all non-tree edges 
$e \in E_G \setminus E_T$ into two edges, respectively referred to $e$, and $e'$ (see Figure~\ref{fig:subdivision}). 
The spanning tree $T'$ of $G'$ is also obtained by adding $e$ to 
the tree $T$. This input transformation naturally defines an injective (but not onto) mapping $\sigma : E_G \to E_{T'}$,
where every edge $e \in E_G$ is mapped to the corresponding edge in $T'$ with the same name.
A query $(s, t, F)$ for $G$ is also naturally interpreted to 
the query $(s, t, \{\sigma(e) | e \in F\})$ for $G'$. It is easy to see that $s$ and $t$ are connected in $G - F$ if and only 
if they are connected in $G' - \{\sigma(e) | e \in F\}$. Hence the following proposition obviously holds:

%Hence any tree edge $f$-labeling scheme for $G'$ and $T'$ inherently works as an $f$-FTC labeling scheme for $G$. 
%The following proposition is obvious:
\begin{proposition} \label{prop:tree-edge-reduction}
Let $G$ and $T$ be the input graph and its rooted spanning tree, and $G'$, $T'$, and $\sigma$ 
be the graphs and the mapping as defined above. Assume any tree edge $f$-FTC labeling scheme $(\TClabel_{G', f}, \TCdec_{f})$ for $G'$ and $T'$. Then we define the labeling function 
$\Clabel_{G, f}$ for $G$ as follows:
\begin{gather*}
\Clabel_{G, f}(x) = 
\begin{cases}
\TClabel_{G', f}(x) & \text{if $x \in V_G \cup E_T$} \\
\TClabel_{G', f}(\sigma(x)) & otherwise.
\end{cases}.
%\Cdec_f((s, t, F)) = \TCdec_f((s, t, \{\sigma(e) \mid e \in F\})).
\end{gather*}
Then $(\Clabel_{G, f}, \TCdec_f)$ is a $f$-FTC labeling scheme for $G$.
\end{proposition}
 
\begin{figure*}[t]
	\begin{center}
		\includegraphics[clip, scale=0.3]{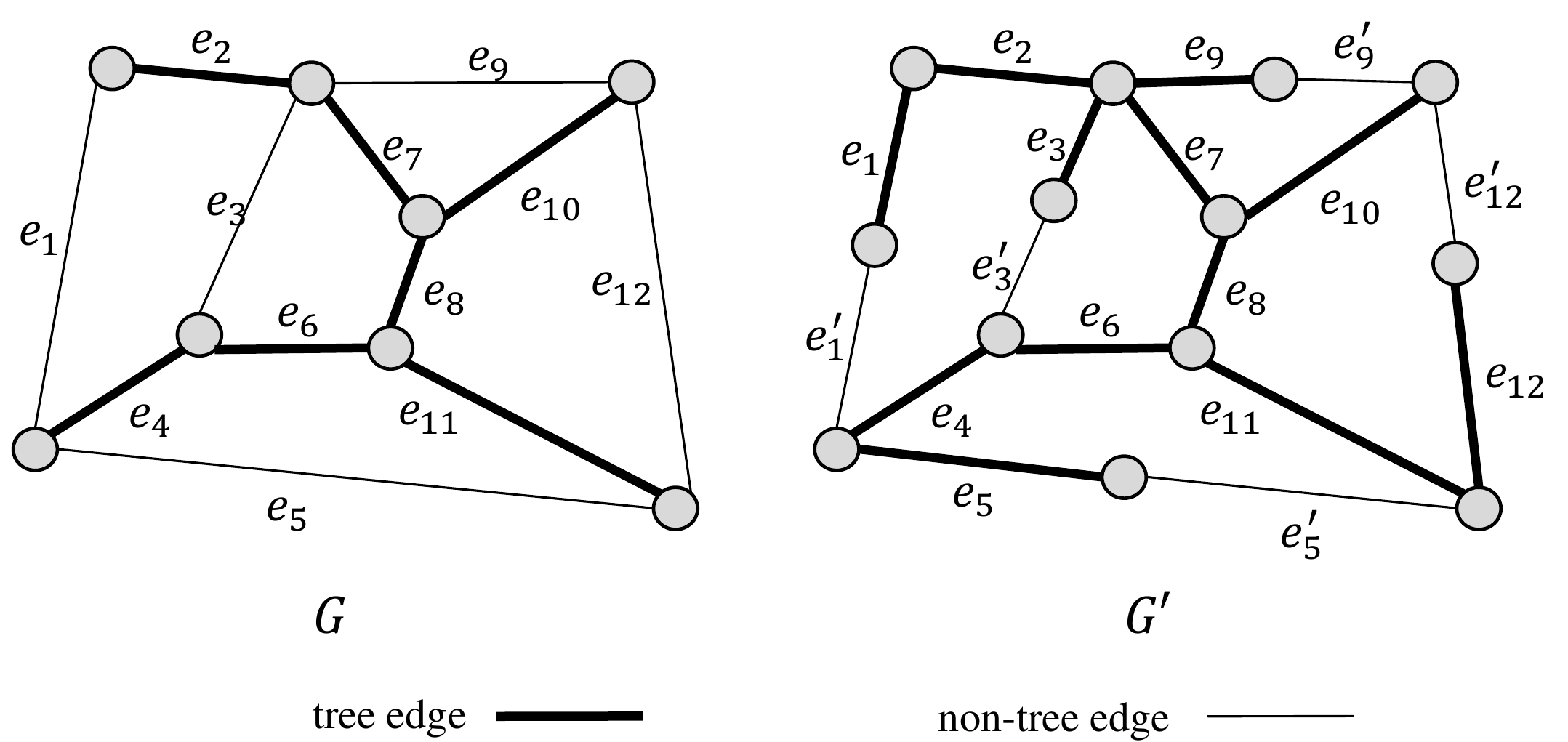}
	\end{center}
	\caption{Auxiliary graph $G'$}
	\label{fig:subdivision}
\end{figure*}

Combining Lemma~\ref{lma:framework} and Proposition \ref{prop:tree-edge-reduction}, we obtain the following
corollary:

\begin{corollary} \label{corol:framework}
Assume any deterministic $\Scal_{f, T}$-outdetect labeling scheme $(\Olabel_G, \Odec)$ of label size $\alpha$ whose decoding time is $\beta$. Then there exists a deterministic $f$-FTC labeling scheme of 
$(\alpha + O(\log n))$-bit label size and $O(|F|(\beta + \log |F|)$ decoding time, where $F$ is a set of faulty edges given by a query.

\end{corollary}

\section{Technical Outline of Our Approach}
\label{sec:ourApproach}

\subsection{Obstacles in De-Randomization}

Corollary~\ref{corol:framework} implies that the difficulty of de-randomization lies only at the implementation of the deterministic $\Scal_{f, T}$-outdetect labeling scheme. The known 
$\Scal$-outdetect labeling 
scheme based on the graph sketch includes two major points reling on random bits, which are 
summarized as follows:
\begin{itemize}
\item The first is at the computation of vertex labels. Let $I_G(v)$ be the set of incident edges 
of $v$ in $G$. The graph sketch first prepares
some function $g : \Ecal \to \{0, 1\}^k$, where $\Ecal$ is the edge ID domain and $k$ is the label length, and define the label 
$\Olabel_G(v)$ of vertex $v$ as the bitwise XOR sum of $g(e)$ for all the edges $e \in I_G(v)$.
When computing $\bigoplus_{v \in S} \Olabel_G(v)$ for a given subset $S \subseteq V_G$, the value $g(e)$ 
for any $e$ lying at the inside of $S$ are canceled out because the term $g(e)$ appears 
exactly twice in the sum $\bigoplus_{v \in S} \Olabel_G(v) = \bigoplus_{v \in S} 
(\bigoplus_{e \in I_G(v)} g(e))$. That is, $\bigoplus_{v \in S} \Olabel_G(v) = 
\bigoplus_{e \in \Cutset_G(S)} g(e)$ holds. For clarifying the essence of the first point, 
we consider the simple case such that $|\Cutset_G(S)| = 1$ holds (the general case is 
addressed in the second point). In this case, 
$\bigoplus_{v \in S} \Olabel_G(v) = g(e)$ obviously holds for the unique outgoing edge $e$ of $S$.
Hence one can extract the outgoing edge ID from $\bigoplus_{v \in S} \Olabel_G(v)$, 
provided that there exists a way of computing the inverse $g^{-1}$. However, if $g$ is 
not well-designed, some subset $S' \in \Scal_{f, T}$ which does not have $e$ as an outgoing 
edge might accidentally satisfy $\bigoplus_{v \in S'} \Olabel_G(v) = g(e)$. Then, $e$ is wrongly
detected as an outgoing edge of $S'$. To avoid it, the graph sketch needs 
to guarantee that $\bigoplus_{v \in S'} \Olabel_G(v)$ becomes different from the value $g(e)$ of 
any edge $e \in E_G$ if $|\Cutset_G(S')| \neq 1$. The first point of utilizing random bits is 
to attain this condition by taking a random hash function $g$. 

\item As explained above, the graph sketch provides the $\Scal$-outdetect labeling scheme working only 
for $S \in \Scal$ satisfying $|\Cutset_G(S)| = 1$. To cover the case of $\Cutset_G(S) > 1$, 
the original scheme prepares the collection $\Gcal$ of spanning subgraphs of $G$ such that  for any $S \in \Scal$ 
there exists a corresponding $H \in \Gcal$ which satisfies $\Cutset_{H}(S) = 1$. The label to vertex $v$ is then obtained by 
the concatenation of $\Olabel_{H}(v)$ for all $H \in \Gcal$. 
Roughly, each spanning subgraph in $\Gcal$ must be (almost everywhere) sparser than the original input
$G$. The construction of each graph in that collection follows a stochastic sampling of edges, which is 
the second point relying on randomization. 
\end{itemize}

The technical highlight of our deterministic scheme is twofold, which respectively resolve the two issues above. 
We explain the outline of each technique in the remainder of this section. The formal argument is provided in
Section~\ref{subsec:outdetect}.

\subsection{First Technique: Deterministic $k$-Threshold Outdetect Labeling Scheme}
\label{subsec:k-threshold}

The first technique is a deterministic function $g$
replacing the random function of the graph sketch, based on the theory of error-correcting codes. To
explain it, we first present a very concise review of coding theory: A \emph{linear 
code} $W$ is a $y$-dimensional linear subspace of $\GFtwo^x$, where $\GFtwo$ is the finite field of two elements (i.e., the element set $\{0, 1\}$ and every calculation is done in modulo 2), $y$ is the length of source data, and 
$x$ is the length of codewords. We abuse $W$ as the set of all codewords. The \emph{minimum distance} of 
a linear code is the minimum  Hamming distance over all pairs of the codewords in $W$. 
In principle, any linear code with minimum distance $k$ can correct any error of less than 
$k/2$ symbols (but it does not necessarily imply that there exists an efficient algorithm of 
correcting errors). One of the standard approaches of correcting errors is the \emph{syndrome decoding} based on
\emph{parity check matrices}. The parity check matrix $C$ of $W$ is the full-rank $x \times (x - y)$ 
matrix satisfying $w \cdot C = 0$ for any codeword $w \in W$. Since the parity check matrix of $W$ 
is uniquely determined from $W$, linear codes are often defined by the corresponding parity check matrices. 
A key property of the parity 
check matrix is that given a codeword with noise $w + \delta$, where $w \in W(A)$ and $\delta$ is a noise vector, 
$(w + \delta) \cdot C = \delta \cdot C$ holds. The \emph{syndrome} of a received (noisy) codeword $w + \delta$ 
is the vector $(w + \delta) \cdot C$, and the syndrome decoding is the process of recovering $\delta$ from 
the syndrome $(w + \delta) \cdot C = \delta \cdot C$. If $\delta$ is correctly recovered, the noiseless codeword $w$ is also 
recovered by adding $\delta$ to the received codeword $w + \delta$.

The background idea of our first technique is as follows: 
We treat $g : \Ecal \to \{0, 1\}^\ell$ as the mapping from $\Ecal$ to $\ell$-dimensional row vectors 
over $\GFtwo$ (where $\ell$ is the label size), and define the $|\Ecal| \times \ell$ matrix 
$C = (c_{e, i})_{e \in \Ecal, i \in [0, \ell-1]}$, where $c_{e, i}$ is the $i$-th bit of $g(e)$.
Let $w(X) = (w_{e})_{e \in \Ecal}$ be the characteristic row vector for $X \subseteq \Ecal$, i.e., $w_e = 1$ 
if $e \in X$, or zero otherwise. Then the following equality holds for any 
$S \subseteq V$:
\begin{align*}
    w(\Cutset_G(S)) \cdot C = \sum_{e \in \Cutset_G(S)} g(e) = \sum_{v \in S} \Olabel_G(v).
\end{align*}
Note that the summation $\sum$ is the sum over $\GFtwo$, and thus that operation is equivalent 
to the bitwise XOR. What we need is the recovery of one non-zero entry in $w(\Cutset_{G}(S))$ 
from the right-side sum. This task can be interpreted into the following scenario: Consider 
the linear code whose parity check matrix is $C$. Then recover the noise vector $w(\Cutset_G(S))$ from 
the syndrome $w(\Cutset_G(S))C = \sum_{v \in S} \Olabel_G(v)$. If the linear
code defined by $C$ has a minimum distance $k > 0$, one can obtain the complete recovery of 
$w(\Cutset_G(S))$ for any $S$ satisfying $|\Cutset_G(S)| < k/2$. Since fixing $C$ implies fixing 
$g$ (and thus the labeling function $\Olabel_G$), one can obtain the deterministic function $g$ from 
the parity check matrix of any linear code. We show that \emph{Reed-Solomon code} nicely fits our 
objective\footnote{Precisely, Reed-Solomon code is a non-binary code. Hence we need to generalize
the argument above slightly, from $\GFtwo$ to any general finite field of characteristic two. The detailed formalism is given in Section~\ref{subsec:specification}.}, which provides the outdetect labeling 
of $O(k\log n)$ bits supporting the detection of \emph{all} outgoing edges of a given subset $S \in 2^{V_G}$ in $O(k^2)$ time
if $|\Cutset_G(S)| \leq k$ holds. We refer to such a scheme as the \emph{$k$-threshold outdetect 
labeling scheme} hereafter.
Let $(\RSlabel{k}_{H}, \RSdec{k})$ be the $k$-threshold outdetect labeling scheme for $H \subseteq G$ 
defined by the $|\Ecal| \times 2k$ parity check matrix of Reed-Solomon code. It satisfies 
the following properties:
\begin{proposition} \label{prop:sftk-outdetect}
\sloppy{
The $k$-threshold outdetect labeling scheme $(\RSlabel{k}_{H}, \RSdec{k})$ satisfies the following 
conditions:
\begin{itemize}
\item The label size is $O(k \log n)$ bits. 
\item The time taken to assign the labels $\RSlabel{k}_H(v)$ to all vertices $v \in V_H$ is $O(mk)$. 
The time of computing $\RSdec{k}(\RSlabel{k}_{H}(S))$ for given $\RSlabel{k}_{H}(S)$ is always
bounded by $O(k^2)$.
\item Given $\RSlabel{k}_{H}(S)$, the output of the decoding function is the IDs of all edges in $\Cutset_H(S)$ if $|\Cutset_H(S)| \leq k$ holds. If $|\Cutset_H(S)| > k$, the returned value is unspecified. That is, an arbitrary value can be returned. 
\end{itemize}
%All time complexities are measured in the standard word-RAM model. 
}
\end{proposition}

In contrast with the single edge detection capability of the original graph sketch, it is a great 
advantage that our technique admits the detection of at most $k$ outgoing edges, which made 
the construction of the collection $\Gcal$ much easier. It suffices 
to construct the sparsification hierarchy $E_{G - E_T} = E_0 \supseteq E_1 \supseteq E_2 \supseteq, \dots, 
\supseteq E_h = \emptyset$ for $h = O(\log n)$ such that every $S \in \Scal_{f, T}$ satisfying $\Cutset_{G}(S) \neq \emptyset$ admits a graph $G_i = (V, E_i)$ satisfying $0 < |\Cutset_{G_i}(S)| \leq k$. We define such a hierarchy as a \emph{$(\Scal, k)$-good hierarchy}:
\begin{definition} \label{def:goodness}
Let $\Scal \subseteq 2^{V_{G}}$ and a $k$ be a positive integer. A \emph{$(\Scal, k)$-good hierarchy} of $E_G - E_T$ 
is the hierarchical edge set $E_G - E_T = E_0 \supseteq E_1 \supseteq E_2 \supseteq, \dots, \supseteq E_h = \emptyset$ satisfying the following conditions
\begin{itemize}
    \item The subset $E_{i+1} \subseteq E_i$ is a constant fraction size\footnote{We 
use the statement ``a subset $X' \subseteq X$ has a constant fraction size (of $X$)'' to mean $|X'| \leq (1 - c)|X|$ for 
some constant $c > 0$.} for any $i \in [0, h-1]$. Note that this condition inherently deduces the property of $h = O(\log n)$.
    \item For any $S \in \Scal$ such that $|\Cutset_{E_i}(S)| > k$, 
    $|\Cutset_{E_{i+1}}(S)| > 0$ holds.
\end{itemize}
\end{definition}
The collection of $k$-threshold outdetect labeling schemes for all $G_i = (V_G, E_i)$ forms a 
$\Scal_{f, T}$-outdetect labeling scheme for $G - E_T$ with $O(k \log^2 n)$-bit label size 
and $O(k^2 \log n)$ decoding time. The decoding process tries to obtain the edge(s)
in $\Cutset_{G_i}(S)$ in the decreasing order of $i$. For the largest $i$ such 
that $\Cutset_{G_i}(S)$ is non-empty, the corresponding $k$-threshold outdetect labeling 
returns a subset of $\Cutset_{G}(S)$ correctly. Formally, the following lemma holds:

\begin{lemma} \label{lma:outdetect}
Assume that any $k$-threshold outdetect labeling scheme $(\hOlabel_H, \hOdec)$ of label size $\alpha$ and 
query processing time $\beta$ is available. If there exists an algorithm of constructing a $(\Scal, k)$-good 
hierarchy for $\Scal \subseteq 2^{V_{G}}$ and $E_G - E_T$, there exists an $\Scal$-outdetect labeling scheme 
$(\Olabel_{G - E_T}, \Odec)$ for $G - E_T$ whose label size is $O(\alpha \log n)$ bits and query processing
time is $O(\beta \log n)$. 
\end{lemma}

\subsection{Second Technique: Deterministic Construction of $(\Scal_{f, T}, k)$-good Hierarchy}
\label{subsec:sparsification}

A crucial requirement of the logarithmic-depth hierarchy $E_0 \supseteq E_1 \supseteq E_2 \supseteq, \dots, \supseteq E_h$ is to avoid the set $S \in \Scal_{f, T}$ such that $|\Cutset_{E_{i+1}}(S)| = 0$ holds 
despite $|\Cutset_{E_i}(S)| > k$ and to make $|E_{i+1}|$ substantially smaller than $|E_i|$. 
It implies that $E_{i+1} \subseteq E_i$ must be a hitting set of the family of edge sets $\Zcal_{i, f, k} = \{\Cutset_{E_i}(S) \mid S \in \Scal_{f, T}, |\Cutset_{E_i}(S)| > k\}$ of a constant fraction size. Allowing randomization, it suffices 
to construct $E_{i+1}$ by the independent edge sub-sampling from $E_i$ with probability $1/2$. 
Such a construction satisfies the desired property for $k = O(f \log n)$ with 
high probability (see the appendix~\ref{appendix:randomized} for details). 
While the standard greedy algorithm can deterministically construct the hitting set with 
the same guarantee, such an approach is not tractable because the size of $\Zcal_{i, f, k}$ could be super-polynomial, 
i.e., $|\Zcal_{i,f, k}| = \Theta(|\Scal_{f, T}|) = \Theta(n^f)$ can hold. The second key 
ingredient of our construction is to provide a polynomial-time deterministic algorithm of constructing the hitting set for 
$k = \tilde{O}(f^2)$ through the geometric representation based on the \emph{Euler-tour structure} 
by Duan and Pettie~\cite{DP10}. In this structure, 
each undirected edge $e$ in $T$ is replaced by two directed edges with opposite orientations. We refer to 
the tree $T$ after the replacement as $\vec{T}$, and extend the definition of $\Cutset_T(S)$ into the directed case $\Cutset_{\vec{T}}(S)$, which consists of all the directed edges obtained by the replacement of an edge in $\Cutset_T(S)$. All the edges in $\vec{T}$ are ordered by any Euler tour $\Et$ of $\vec{T}$ starting from 
the root $r$, and each vertex in the tree is assigned with the smallest order of the incident in-edge (i.e. the edge coming from its parent) as its one-dimensional coordinate in the range $[1, 2n - 2]$. 
We denote the one-dimensional coordinate of a vertex $v \in V_{\vec{T}}$ by $c(v)$. Then, one can map each non-tree edge
$e = (u, v)$ into the 2D-point $(c(u), c(v))$ in the range $[1, 2n - 2] \times [1, 2n - 2]$ (assuming the $x$-coordinate is always smaller than the $y$-coordinate to make the mapping well-defined). An example 
of the geometric representation for the instance of Figure~\ref{fig:subdivision} is presented in Figure~\ref{fig:cutgeometry}. For any $a \in [1, 2n - 2]$ and $z \in \{x, y\}$, let $\Hp(z, a)$ be the 
axis-aligned halfspace defined by $z \geq a$. Then it is observed that the point set $\Cutset_{E_i}(S)$ for 
any $S \in \Scal_{f, T}$ lies in the symmetric difference of at most $4f$ axis-aligned 
halfspaces. More precisely, the following lemma holds:
\begin{lemma} \label{lma:cutspace}
For any vertex subset $S \subseteq V_G$ and edge subset $E' \subseteq E_{\Gext}$, the following equality holds.
\begin{align*}
\Cutset_{E'}(S) = E' \cap \left(\symdiff_{e \in \Cutset_{\vec{T}}(S), z \in \{x, y\}} \  {\Hp(z, c(e))} \right).
\end{align*}
where $\symdiff$ represents the symmetric difference of sets. 
\end{lemma}
This lemma implies that every $\Cutset_{E_i}(S)$ for 
$S \in \Scal_{f, T}$ is associated with a ``checkered shape'' in the plane with at most $2f$ vertical (or horizontal) alternations. In Figure~\ref{fig:cutgeometry}, the region colored by white corresponds to the outgoing edges of $S$ 
such that $\Cutset_{\vec{T}}(S)$ consists of two directed edges with numbers 3 and 18. Then the problem of constructing 
the hitting set is seen as the construction of \emph{$\epsilon$-nets}~\cite{HW87}.
\begin{definition}[$\epsilon$-nets]
Let $\Zcal$ be a class of geometric shapes (e.g., rectangles or disks) in some space, and $X$ be a set of points in the space. 
An \emph{$\epsilon$-net} for $(X, \mathcal{Z})$ is a subset $X' \subseteq X$ such that for any $Z \in \Zcal$, $Z \cap X' \neq 
\emptyset$ holds if $|Z \cap X| \geq \epsilon|X|$. 
\end{definition}

Let us define the class $\Hcal_{q}$ which 
consists of all the shapes formed by the symmetric difference of at most $q$ horizontal halfspaces $\Hp(y, a_0), \Hp(y, a_1), \dots, \Hp(y, a_{q'-1})$ ($q' \leq q$, $a_i \in [1, 2n - 2]$) and the corresponding vertical halfspaces 
$\Hp(x, a_0), \Hp(x, a_1), \dots, \Hp(x, a_{q'-1})$. Recalling that the construction of 
$E_{i+1}$ is equivalent to the computation of a hitting set of a constant fraction size 
for the family $\Zcal_{i, f, k} = \{\Cutset_{E_i}(S) \mid S \in \Scal_{f, T}, |\Cutset_{E_i}(S)| > k\}$, our goal is to
construct the hitting set of a constant fraction size for $\{Z \cap E_{i} \mid Z \in \Hcal_{2f}, |Z \cap E_i| \geq k\}$, 
i.e., to construct an $\epsilon$-net of a constant fraction size 
for $(E_i, \Hcal_{2f})$ and $\epsilon = k/|E_i|$.

\begin{figure*}[t]
	\begin{center}
		\includegraphics[clip, scale=0.4]{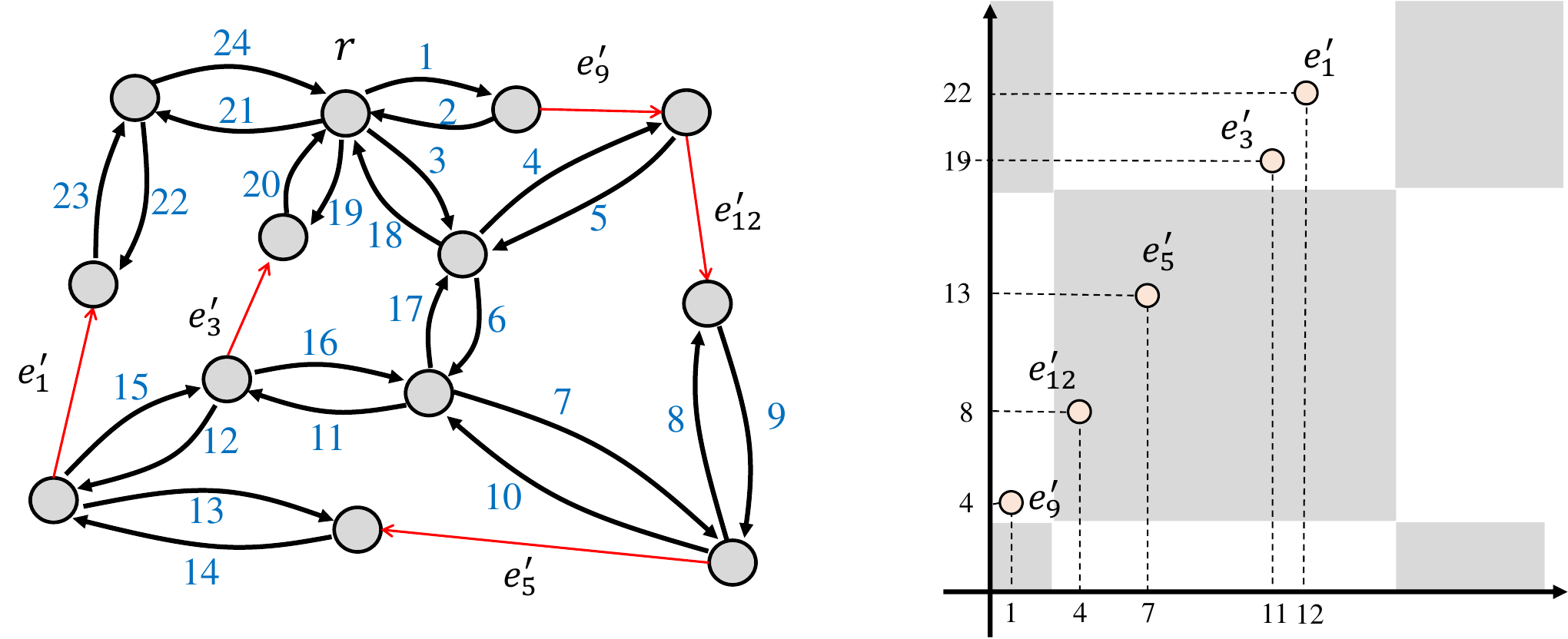}
	\end{center}
	\caption{The geometric interpretation of cutsets. The blue number is the ordering of the directed tree edges by an Euler tour. The non-tree edges $e'_1$, $e'_3$, $e'_5$, $e'_9$ and $e'_{12}$ are respectively mapped into the points shown in the right-side coordinate system.  }
	\label{fig:cutgeometry}
\end{figure*}

While there are a few deterministic polynomial-time algorithms of constructing nearly-optimal $\epsilon$-nets for a given class 
of shapes~\cite{Matousek96,CM96}, their running times exponentially depend on the VC dimension of the given class.
The VC dimension of $\Hcal_{2f}$ is $\Omega(f)$, and thus those algorithms cannot be applied\footnote{More precisely, they takes $\Omega(1/\epsilon)^d$ time for the class with VC dimension $d$. This would efficiently work if $1/\epsilon$ is small, but in our use $1/\epsilon$ is roughly close to $m/f$, and thus they are not tractable for $f = \omega(1)$.}. 
To circumvent this issue, we regard any shape in $\Hcal_{2f}$ as the union of $(2f + 1)^2/2$ disjoint axis-aligned rectangles. 
For any $H \subseteq \Hcal_{2f}$ containing at least $\gamma (2f + 1)^2/2$ points ($\gamma \geq 1$), there exists at least 
one axis-aligned rectangle as a subset of $H$ which contains at least $\gamma$ points. Hence the construction of 
an $O(\gamma / |E_i|)$-net of a constant fraction size for all axis-aligned rectangles deduces an $O(\gamma f^2/ |E_i|)$-net 
of a constant fraction size for $\Hcal_{2f}$. In other words, we can obtain a deterministic polynomial-time algorithm of 
constructing $(\Scal_{f, T}, O(\gamma f^2))$-good hierarchy from any deterministic 
polynomial-time algorithm of constructing a $(\gamma/ N)$-net of a constant fraction size for $N$ points and 
all axis-aligned rectangles.

Since the VC-dimension of axis-aligned rectangles is a constant, 
it is possible to use the general de-randomization technique as stated above. The optimal $\epsilon$-net for $N$ points and axis-aligned rectangles is
of size $O(\log\log N / \epsilon)$ (i.e., $O(\log\log N / N)$-net of a constant fraction size), 
which is known to be deterministically constructed~\cite{MDG18}. However, the construction time takes a high-exponent polynomial. 
Hence we also present a simpler alternative construction which provides a 
$O(\log N / N)$-net of a constant fraction size for axis-aligned rectangles in a near linear time. In summary, 
we obtain the following lemma:
\begin{lemma}[Partly By Moustafa, Dutta, and Ghosh~\cite{MDG18}] \label{lma:e-net}
There exist two deterministic algorithms of constructing an $O(\gamma / N)$-net of a constant fraction size 
for any $N$ points and all axis-aligned rectangles, each of which attains the following performance guarantee.
\begin{itemize}
\item $\gamma = \log\log N$ and the construction time is $\Poly(N)$.
\item $\gamma = \log N$ and the construction time is $\tilde{O}(N)$.
\end{itemize}
\end{lemma}
By the argument above, we obtain the following lemma.
\begin{lemma} \label{lma:goodhierarchy}
There exist two deterministic algorithms respectively constructing a $(\Scal_{f, T}, k)$-good hierarchy with 
the following performance guarantees:
\begin{itemize}
\item $k = O(f^2 \log n)$ and the construction time is $\tilde{O}(m)$.
\item $k = O(f^2 \log \log n)$ and the construction time is $\Poly(m)$.
\end{itemize}

\end{lemma}

\section{Wrap-Up}
\label{sec:wrapup}

We summarize how all the components are combined into the $f$-FTC labeling scheme. We present below the case of 
the $f$-FTC labeling scheme of label size $O(f^2 \log^3 n)$. Yet another scheme of label size $O(f^2 (\log^2) \log\log n)$ is  
constructed in the same way.
Consider any input graph $G$ of $n$ vertices and $m$ edges. The whole construction algorithm works as the following steps:
\begin{enumerate}
\item Construct any spanning tree $T$ of $G$, and transform $G$ and $T$ into the auxiliary graph $G'$ and its spanning 
tree $T'$ explained in Section~\ref{subsec:transformation}. Note that the graph $G'$ satisfies 
$|V_{G'}| = O(m)$ and $|E_{G'}| = O(m)$. 
\item Utilizing the algorithm of Lemma~\ref{lma:goodhierarchy},
construct a $(\Scal_{f, T}, c f^2 \log n)$-good hierarchy 
$E_0 \supseteq E_1 \supseteq E_2 \supseteq, \dots, \supseteq E_h$ for $E_{G'} - E_{T'}$, where 
$c$ is a hidden constant. The construction time is $\tilde{O}(|E_{G'}|) =  O(m)$. 
\item Let $G_i = (V_{G'}, E_i)$ $(0 \leq i \leq h)$. Construct
$\RSlabel{c f^2 \log n}_{G_i}$ for all $i$. By the construction of step 2 and Lemma~\ref{lma:outdetect},
we obtain the $\Scal_{f, T}$-outdetect labeling scheme of $O(f^2 \log^3 n)$-bit
label size. The construction time is in $\tilde{O}(mkh) = \tilde{O}(mf^2)$, and 
the decoding time is $\tilde{O}(f^4)$ 
due to Proposition~\ref{prop:sftk-outdetect} and Lemma~\ref{lma:outdetect}. 
\item Construct the ancestry labeling scheme for $T'$ by the algorithm of \cite{KNR92}, which takes $O(m)$ time.
\item By Lemma~\ref{lma:framework}, the labels constructed in the steps 3 and 4 form a tree edge $f$-FTC labeling scheme 
for $G'$ and $T'$. Then we also obtain the $f$-FTC labeling scheme for $G$ by Corollary~\ref{corol:framework}.
\end{enumerate}
Finally we have the following theorem.

\begin{theorem} \label{thm:nonadaptivemaintheorem}
There exist two deterministic $f$-FTC labeling schemes for any graph $G$ of $n$ vertices and $m$ edges
which respectively attain the following bounds:
\begin{itemize}
    \item The label size is $O(\log n)$ bits per vertex, and $O(f^2 (\log^2 n) \log\log n)$ bits 
    per edge. The query processing time is $\tilde{O}(|F|f^4)$, where $F$ is the set of queried edges 
    satisfying $|F| \leq f$. The construction time is polynomial of $m$.
    \item The label size is $O(\log n)$ bits per vertex, and $O(f^2 \log^3 n)$ bits 
    per edge. The query processing time is $\tilde{O}(|F|f^4)$. The construction time is near linear, i.e., $\tilde{O}(mf^2)$.
\end{itemize}
\end{theorem}

This theorem is a weaker form of Theorem~\ref{thm:maintheorem}. In the next section 
we present how one can improve this to Theorem~\ref{thm:maintheorem} claiming faster query processing time.

\section{Improving Query Processing Time}
\label{sec:improvement}

The algorithmic idea of
this improvement is twofold: The first idea attains the \emph{adaptiveness} of $\Scal_{f, T}$-outdetect labeling scheme, i.e., 
to get rid of the dependency on $f$ in the decoding time. There is a simple technique of transforming any $\Scal_{f, T}$-outdetect 
labeling scheme with $\tilde{O}(f^c)$ decoding time into the one with $\tilde{O}(|\Cutset_T(S)|^c)$ decoding time 
for any given query $S \in \Scal_{f, T}$: Instead of single labeling, we prepare the multiple instances of 
the (non-adaptive) $\Scal_{f', T}$-outdetect labeling scheme for $f' = 2, 4, \dots, f$.
If the original labeling scheme has a $\Omega(f)$-bit label size, this transformation does not cause any
asymptotic blow-up of label size. Assume that a query $S \subseteq \Scal_{f, T}$ is given. 
Since $S \in \Scal_{|\Cutset_T(S)|, T}$ necessarily holds, the adaptive scheme can find 
the outgoing edge of $S$ by utilizing the $\Scal_{f', T}$-outdetect labeling scheme for $f'$ such that 
$f'/2 < |\Cutset_T(S)| \leq f'$ holds, which runs in $\tilde{O(}(|\Cutset_T(S)|^c)$ time
\footnote{In reality, this transformation is not necessary if we utilize our deterministic $\Scal_{f, T}$-outdetect 
labeling scheme based on the Reed-Solomon code. More precisely, it inherently admits the adaptive decoding without
any modification of the label construction. See the appendix \ref{appendix:adaptiveRS} for details.}.

The second idea is to utilize the adaptive scheme for further acceleration of the decoding time.
In processing the query of $(s, t, F)$, every query $S \subseteq V_G$ issued to 
the $\Scal_{f, T}$-outdetect labeling scheme necessarily belongs to $\Scal_{|F|, T}$. 
Hence The adaptive decoding of the $\Scal_{f, T}$-outdetect labeling 
scheme always runs in $\tilde{O}(|F|^4)$ time for the deterministic cases, 
and in $\tilde{O}(|F|^2)$ time for the randomized case. 
Since the decoding time of the tree edge $f$-FTC labeling scheme is dominated by $|F|$ queries to the 
$\Scal_{f, T}$-outdetect labeling scheme, the straightforward implementation respectively results 
in the decoding time of $\tilde{O}(|F|^5)$ and $\tilde{O}(|F|^3)$. We shave off this extra 
$|F|$ factor by a simple refinement of the decoding process of the tree edge $f$-FTC labeling scheme: In the refined process, 
the outgoing edge detection for merging fragments is always applied to the fragment $S$ such that $|\Cutset_T(S)|$ is 
the smallest of all the fragments currently managed, while the original process always applies it to the fragment with $s$. 
By a careful analysis, we obtain the following lemma:
\begin{lemma}\label{lma:improveddecoding}
Assume that there exists a $\Scal_{f, T}$-outdetect labeling scheme of label size $\alpha = \tilde{O}(f^b)$ and 
decoding time $\beta = \tilde{O}(f^c)$. Then there exists a $f$-FTC labeling scheme of $O(\alpha + \log n)$-bit label size 
and $\tilde{O}(|F|^{b+1} + |F|^c)$ decoding time. The resultant $f$-FTC labeling scheme is deterministic if the corresponding
$\Scal_{f,T}$-outdetect labeling scheme is deterministic.
\end{lemma}
\sloppy{
The three $\Scal_{f, T}$-outdetect labeling schemes we presented in this paper (including the randomized case) 
attain $(\alpha, \beta) = 
(O(f^2 (\log^2 n) \log\log n), \tilde{O}(f^4)), (O(f^2 \log^3 n), \tilde{O}(f^4))$, and $O(f \log^3 n), \tilde{O}(f^2))$.
By this lemma, they respectively deduce the schemes as claimed in Table~\ref{tab:comparison}. 
}

\section{Technical Details}
\label{sec:details}

\subsection{Formal Specification of Labeling Schemes} \label{subsec:specification}

\paragraph{Fault-Tolerant Connectivity Labeling}
A \emph{$f$-fault-tolerant connectivity labeling scheme} ($f$-FTC labeling scheme) for a given input graph 
$G$ consists of a labeling function $\Clabel_{G, f} : V_G \cup E_G \to \Allwords$ and 
a universal decoding function $\Cdec_{f} : \Allwords \times \Allwords \times \Allwords \to \{0, 1\}$.
For an edge subset $F = \{e_1, e_2, \dots, e_{|F|}\} \subseteq E_G$, let $\Clabel_{G, f}(F)$ be 
the concatenation of the labels $\Clabel_{G, f}(e_1) \circ \Clabel_{G, f}(e_2) \circ \dots \circ 
\Clabel_{G, f}(e_{|F|})$ in an arbitrary order.
For a query $(s, t, F)$ of $s, t \in V_G$ and an edge subset $F \subseteq E_G$ of cardinality at most $f$, 
the decoder function returns the $s$-$t$ connectivity in $G - F$ by giving
the labels of $s$, $t$ and all the edges in $F$, i.e., 
$\Cdec$ satisfies that $\Cdec(\Clabel_{G}(s), \Clabel_{G}(t), \Clabel_{G}(F)) = 1$ 
if and only if $s$ and $t$ are connected in $G - F$. The \emph{label size} of the scheme is defined as the maximum length of the labels assigned to vertices and edges. 

\paragraph{Ancestry labeling}
Given any rooted tree $T$, this labeling scheme assigns 
all vertices with the labels such that the ancestor-descendant relationship between 
any two vertices is determined only from their labels. More precisely, the ancestry labeling scheme for $T$ 
consists of a labeling function $\Alabel_{T} : V_T \to \Allwords$ and a universal decoding 
function $\Adec : \Allwords \times \Allwords \to \{-1, 0, 1\}$ (not dependent on $T$). 
It determines if two given vertices $x, y \in V_T$ have the ancestor-descendant relationship or not 
from their labels, i.e., $\Adec(\Alabel_T(x), \Alabel_T(y)) = 1$ 
if $x$ is an ancestor of $y$, $-1$ if $y$ is an ancestor of $x$, or $0$ otherwise (including the case 
of $x = y)$. The following lemma is well-known:

\begin{lemma}[Kannan, Naor, and Rudich~\cite{KNR92}] \label{lma:anclabel}
Let $T$ be a rooted tree of $n$ vertices. There exists a deterministic 
ancestry labeling scheme with label size of $O(\log n)$ bits. Computing $\Alabel_{T}(x)$ for all 
$x \in V_T$ takes $O(n)$ time, and $\Adec(\Alabel_{T}(x), \Alabel_{T}(y))$ for each $x, y \in V_T$ takes 
$O(1)$ time. The labeling function $\Alabel_T$ is injective for any $T$, i.e., a unique label assignment.
\end{lemma}

\paragraph{$\Scal$-Outdetect Labeling}
Let $\Field$ be any finite field of characteristic two whose addition and multiplication operators
are respectively denoted by ``$+$'' and ``$\cdot$''~\footnote{A field $\Field$ has characteristic two if and only if any element $x \in \Field$ satisfies $x + x = 0$, or equivalently $x = x^{-1}$ (where $0$ is the unit element and $x^{-1}$ is the inverse element of $x$).}, and $\Ecal$ be the domain of unique
edge IDs not depending on $G$\footnote{Since the size of this domain inherently depends on $n$, it should be 
defined precisely as $\Ecal_n$, which is the universal domain valid for all the graphs with at most $n$ 
vertices. But we intentionally omit such a dependency for avoiding non-essential complication.}.
An $\Scal$-outdetect labeling scheme consists of a vertex labeling function $\Olabel_{G} : V_G \to \Field^\ell$ and a universal decoding function $\Odec: \Field^\ell \to \Ecal$, where $\ell$ is a positive
integer representing the size of labels. It must satisfy the following two conditions:
\begin{itemize}
    \item For any $S \subseteq \Scal$, $\Odec(\sum_{v \in S} \Olabel_G(v))$ returns the ID of an outgoing edge of $S$ in $G$ if $\Cutset_G(S)$ is nonempty.
    \item If $\Cutset_G(S) = \emptyset$, $\Odec(\sum_{v \in S} \Olabel_G(v))$ returns formal zero, which is a
    special value in $\Ecal$ never assigned to actual edges.  
\end{itemize}
For short, we use the notation $\Olabel_G(S)$ to represent $\sum_{v \in S} \Olabel_G(v)$ in the following argument.
By definition, $\Odec(\Olabel_G(V_G)) = 0$ must hold. The \emph{$k$-threshold outdetect labeling scheme} for $G$ is 
a restricted variant of $2^{V_G}$-outdetect labeling scheme, which guarantees that $\Odec(\Olabel_G(S))$ returns 
an outgoing edge of $S$ only if $0 < |\Cutset_G(S)| \leq k$ holds, or returns zero if $|\Cutset_G(S)| = 0$. 
In the case of $|\Cutset_G(S)| > k$, the returned value is undefined, i.e., an arbitrary value is returned.

\subsection{Proof of Lemma~\ref{lma:framework}} \label{subsec:treeFTC}

Let $\Gext = G - E_T$ for short. The tree edge $f$-FTC labeling scheme is implemented by any 
$\Scal_{f, T}$-outdetect labeling scheme for $\Gext$ and any ancestry labeling scheme
for $T$. Let $(\Alabel_T, \Adec)$ be the ancestry labeling scheme. Without loss of generality,
we assume that $\Alabel_T(v)$ for all $v$ has a fixed bit length $p = O(\log n)$. We assign
the edge ID $\Alabel_T(u) \circ \Alabel_T(v) \in \{0, 1\}^{2p}$ to each edge $(u, v) \in E_{\Gext}$.
The uniqueness of the edge ID follows the uniqueness of the ancestry labeling guaranteed in 
Lemma~\ref{lma:anclabel}. For the edge ID domain $\Ecal = \{0, 1\}^{2p}$, we construct the 
$\Scal_{f, T}$-outdetect labeling scheme $(\Olabel_{\Gext}, \Odec)$. The labeling function $\Clabel_G(v)$ 
of our tree edge $f$-FTC labeling scheme is defined as follows:
\begin{itemize}
    \item $\Clabel_G(v) = \Alabel_{T}(v)$.
    \item For any $e = (u, v) \in E_T$, $\Clabel_G(e) = \Alabel_{T}(u) \circ \Alabel_{T}(v) \circ \Olabel_{\Gext}(V_{T(e)})$. 
    Recall that any tree edge $f$-FTC labeling scheme does not have to assign labels to non-tree edges. 
\end{itemize}
We see how to implement the decoding function. Assume that a query $(s, t, F)$ of $|F| \leq f$ 
is given. Let $\Ccal(F) = \{C_0, C_1, \dots, C_{|F|}\}$ be the collection of the fragments (i.e.,
the vertex subset inducing a connected component of $T - F$). Let $V(F)$ be the set of 
the endpoints of the edges in $F$. By introducing any total order over $\{0, 1\}^{p}$,
we define the ID of $C_i \in \Ccal(F)$ as the maximum ancestry label in $V_{C_i} \cap V(F)$, 
and abuse $C_i$ itself as the ID of $C_i$.  
The \emph{component graph} $H/\Ccal(F)$ of a spanning subgraph $H \subseteq G$ is the multigraph 
obtained from $H$ by contracting each vertex set $C_i$ ($0 \leq i \leq |F|$) into a single vertex 
and then removing self-loops. The following proposition is known:
\begin{proposition}[Claim 3.14 in~\cite{DP21}] \label{prop:identifycomponenttree}
Let $(s, t, F)$ be any given query, and $|F| \leq f$. Then $T / \Ccal(F)$ is computed deterministically 
in $O(|F|\log |F|)$ time. In addition, there exists a deterministic algorithm which is given $\Alabel_T(v)$ 
and returns in $O(\log |F|)$ time the ID of $C_i \in \Ccal(F)$ containing $v$.
\end{proposition}
The proposition above admits the detection of the two connected components in $T - F$ respectively 
containing $s$ and $t$. %If both $s$ and $t$ belong to the same connected component of $\Gext / \Ccal(F)$, 
%$s$ and $t$ are obviously connected in $G - F$, and thus the decoding function returns one. Otherwise,
%one can conclude that $s$ and $t$ are disconnected in $G - F$. 
We assume $s \in C_0$ and $t \in C_1$ without
loss of generality. Starting from $C_0$, the decoding procedure grows the component containing $s$ iteratively
by finding its outgoing edge: Initially, let $S = C_0$. The procedure detects an outgoing edge $(u, v)$ of 
$S$ in $\Gext$, where $u$ is the vertex in $S$ and $v \in C_j$ for some $j \in [0, f']$, until 
no outgoing edge is found. When the edge $(u, v)$ is found, $S$ is updated with $S \cup \{C_j\}$. Obviously, 
$s$ and $t$ are connected in $G - F$ if and only if this process merges $C_1$ with $S$. Throughout this process, $\Cutset_T(S) \subseteq F$ obviously holds and thus $S \in \Scal_{f, T}$ always holds. 
Hence one can use the $\Scal_{f, T}$-outdetect labeling scheme to find an outgoing edge $(u, v)$. 
The primary matter is how to manage $\Olabel_{\Gext}(S)$. 
We resolve it with a technique similar to \cite{DP21} (Claim 3.15). The following proposition holds.

\begin{proposition} \label{prop:computeoutdetectlabel}
For any $X \subseteq V_T$, $\Olabel_{\Gext}(X) = \sum_{e \in \Cutset_T(X)} \Olabel_{\Gext}(V_{T(e)})$ holds.
\end{proposition}

\begin{proof}
Let $k = |\Cutset_T(X)|$. The proof is based on the induction on $k$. (Basis) $k = 0$: 
Then $X = V_G$ holds. Since we have $\Olabel_{\Gext}(V_T) = 0$ by the definition of the $\Scal_{f, T}$-outdetect labeling 
scheme,
the proposition obviously holds. (Inductive step) 
Suppose that the proposition holds for any $X'$ such that $|\Cutset_T(X')| = k$, and consider $X$
such that $|\Cutset_T(X)| = k + 1$ holds. Let $\Cutset_T(X) = \{e_1, e_2, \dots, e_{k+1}\}$. Without loss of 
generality, we assume that there is no descendant of $e_{k+1}$ in $\Cutset_T(X)$, i.e., $T(e_{k+1})$ does
not contain any edge in $\Cutset_T(X)$. Let $e_{k+1} = (u, v)$ where $v$ is the lower vertex of $e_{k+1}$. 
We consider the case of $v \in X$. Let $W = X \setminus V_{T(e_{k+1})}$. Since we have $\Cutset_T(W) = \Cutset_T(X) 
\setminus \{e_{k+1}\}$, $|\Cutset_T(W)| = k$ holds. Thus by the induction hypothesis we obtain $\Olabel_{\Gext}(W) = 
\sum_{e \in \Cutset_T(W)} \Olabel_{\Gext}(V_{T(e)})$. Since $W$ and $V_{T(v)}$ are disjoint by definition, 
we have 
\begin{align*}
\Olabel_{\Gext}(X) &= \Olabel_{\Gext}(W) + \Olabel_{\Gext}(V_{T(e_{k+1})}) \\
& = \sum_{e \in \Cutset_T(W)} \Olabel_{\Gext}(V_{T(e)}) + \Olabel_{\Gext}(V_{T(e_{k+1})}) \\
& = \sum_{e \in \Cutset_T(X)} \Olabel_{\Gext}(V_{T(e)}). 
\end{align*}
The case of $v \not\in X$ is proved similarly. \end{proof}

We prove Lemma~\ref{lma:framework} by the proposition above.

\begin{rlemma}{lma:framework}
Assume any deterministic $\Scal_{f, T}$-outdetect labeling scheme $(\Olabel_G, \Odec)$ of label size $\alpha$ and decoding time $\beta$. Then there exists a deterministic tree edge 
$f$-FTC labeling scheme of $(\alpha + O(\log n))$-bit label size and 
$O(|F|(\beta + \log |F|))$ decoding time, where $F$ is a set of faulty edges given by a query.
\end{rlemma}

\begin{proof}
By Proposition \ref{prop:computeoutdetectlabel}, one can compute the $\Olabel_{\Gext}(C_i)$ for all $i \in [0, |F|]$ 
from the labels of the edges in $F$. Assume an outgoing edge $e = (u, v)$ of $S$ 
($v \not\in S$) is detected by the $\Scal_{f, T}$-outdetect labeling. Since we can obtain the ID $\Alabel_T(u) \circ \Alabel_T(v)$ of $e$, 
the ancestry label $\Alabel_T(v)$ is available. By Proposition~\ref{prop:identifycomponenttree}, it also gives the ID 
of the component in $\Ccal(F)$ which contains $v$. Let us assume $v \in C_j$. When merging $C_j$ into $S$, 
we have known both $\Olabel_{\Gext}(S)$ and $\Olabel_{\Gext}(C_j)$. Thus the label $\Olabel_{\Gext}(S)$ is  
updated by adding $\Olabel_{\Gext}(C_j)$. The query processing time
is spent for $|F|$ times of querying the $\Scal_{f, T}$-outdetect labeling scheme, which takes $O(|F|(\beta + \log |F|))$ time in total.
The initial set-up takes $O(|F|\log |F| + |F|\beta)$ time,
where the term $|F|\log |F|$ is for computing the component graph, and $|F|\beta$ is for computing the $\Scal_{f, T}$-outdetect labels 
of all fragments. The label size is obviously $O(\alpha + \log n)$.
\end{proof}

\subsection{Construction of Deterministic $\Scal_{f, T}$-Outdetect Labeling Scheme} \label{subsec:outdetect}

\begin{rlemma}{lma:outdetect}
Assume that any $k$-threshold outdetect labeling scheme $(\hOlabel_H, \hOdec)$ of label size $\alpha$ and 
query processing time $\beta$ is available. If there exists an algorithm of constructing a $(\Scal, k)$-good 
hierarchy for $\Scal \subseteq 2^{V_{G}}$ and $E_G - E_T$, there exists an $\Scal$-outdetect labeling scheme 
$(\Olabel_{G - E_T}, \Odec)$ for $G - E_T$ whose label size is $O(\alpha \log n)$ bits and query processing
time is $O(\beta \log n)$. 
\end{rlemma}

\begin{proof}
Assume that a $(\Scal, k)$-good hierarchy $E_0 \supseteq E_1 \supseteq E_2 \supseteq, \dots, 
\supseteq E_h$ is obtained. Let $G_i = (V_G, E_i)$. The labeling function $\Olabel_{\Gext}$
is defined as the concatenation of the labels by $\hOlabel_{G_i}$ for all $i \in [0, h]$, 
i.e., $\Olabel_{G - E_T}(v) = \hOlabel_{G_0}(v) \circ \hOlabel_{G_1}(v)
\circ \dots, \circ \hOlabel_{G_h}(v)$ (where $\circ$ is the binary operator of concatenating two 
strings). To detect an outgoing edge of $S \in \Scal$, it suffices to compute
the value of $\hOdec(\bigoplus_{v \in S} \hOlabel_{G_i}(v))$ such that it returns a non-zero value and 
$\hOdec(\bigoplus_{v \in S} \Olabel_{G_j}(v))$ for any $j > i$ returns zero. 
%By the specification of $k$-threshold outdetect labeling scheme, $|E_i \cap \Cutset_{G}(S)| = 0$ holds. 
The condition of the hierarchy implies $0 < |E_i \cap \Cutset_{G}(S)| \leq k$, and thus 
$\hat{\Odec}(\bigoplus_{v \in S} \Olabel_{G_i}(v))$ correctly returns all the outgoing edges of 
$S$ in $E_i \subseteq E_G$. The bounds for the label size and the query processing time are obvious.
\end{proof}

\subsection{$k$-threshold outdetect labeling: Choice of Codes}

Following the construction presented in Section~\ref{subsec:k-threshold}, 
any linear code of the minimum distance $2k$ naturally induces a $k$-threshold outdetect labeling scheme
for any graph $G$. The label size is determined by the number of columns of the parity check matrix $C$
of the code. 
Obtaining both a smaller number of rows and a larger minimum distance is roughly equivalent 
to achieving a good code rate. Hence, as a general principle, any high-rate code would provide a 
good scheme. In addition, we need to care other additional criteria for efficient implementation of 
the outdetect labeling scheme, which are stated as follows:
\begin{itemize}
    \item We use a parity check matrix which has $|\Ecal|$ rows, i.e., the codeword length is $|\Ecal|$. 
    Thus it is not appropriate to use the error-correcting codes whose decoding time depends on the codeword length. Ideally, the decoding time should depend only on the length of the syndrome (i.e., the length 
    of labels).
    \item The construction of the label for an edge $e$ corresponds to the computation of the 
    row vector of $C$ corresponding to $e$. Since the number of rows $|\Ecal|$ could be much larger 
    than the actual number of edges $|E_G|$, computing the whole matrix $C$ can result in slower computation of edge labels to all $e \in E_G$.
    To complete the label assignment in time dependent only on the actual number $|E_G|$ of edges 
    but not on $|\Ecal|$, $C$ must admit efficient ``local'' computation of a specified row. 
\end{itemize}
One of the error-correcting codes addressing the issues above is Reed-Solomon code. Reed-Solomon code is a 
non-binary code whose alphabet is a finite field $\Field$ of order $|\Ecal| + 1$, and the number of rows 
is chosen arbitrarily. Since $|\Ecal|$ is a polynomial of $n$ in our application, each code symbol is encoded
with $O(\log n)$ bits, and the addition and multiplication over $\Field$ takes $O(1)$ time in the standard word-RAM model.
Let $C_{2k}$ be the $|\Ecal| \times 2k$ parity check matrix of Reed-Solomon code. We have the following nice features:
\begin{itemize}
    \item  The minimum distance of the code defined by $C_{2k}$ is equal to $2k$. That is, 
    a $k$-threshold outdetect labeling scheme is deduced from $C_{2k}$. Since 
    each column vector of $C_{2k}$ is encoded by $O(k \log |\Ecal|)$ bits, the label size is 
    $O(k \log n)$ bits.
    \item Let $w$ be any $|\Ecal|$-dimensional vector over $\Field$ which contains at most $k$ nonzero 
    elements. There exists a deterministic algorithm of computing all non-zero elements in $w$ (in 
    the form of the pairs of value and position) from $w \cdot C_{2k}$, which runs in $O(k^2)$ time 
    in the standard word-RAM model~\cite{DORS08}. The recovery of $w$ necessarily succeeds 
    if $w$ contains at most $k$ non-zero elements, but the result can become arbitrary 
    if $w$ contains more than $k$ non-zero elements.
    \item Given any $e \in \Ecal$, the row vector of $C_{2k}$ corresponding to $e$ is deterministically 
    computed in $O(k)$ time in the standard word-RAM model. 
\end{itemize}
Let $(\RSlabel{k}_{H}, \RSdec{k})$ be the $k$-threshold outdetect labeling scheme for $H \subseteq G$ 
defined by the $|\Ecal| \times 2k$ parity check matrix of Reed-Solomon code. The features above
obviously deduces Proposition~\ref{prop:sftk-outdetect}.

\subsection{Deterministic Construction of Good Hierarchy} 
\label{subsec:edgesethierarchy}

We first focus on the proof of Lemma~\ref{lma:cutspace}.
Since $\Cutset_{G'}(S) = \Cutset_{G'}(V_G \setminus S)$ holds for any $S \subseteq V_G$ and any spanning
subgraph $G'$ of $G$, we assume $S$ always contains the root $r$ wlog. 
We use notations $\vec{T}$, $\Et$, $c(v)$, $\Hp(z, a)$, and $\Hcal_{2f}$ as defined 
in Section~\ref{subsec:sparsification}. In addition, we introduce a few additional notations. 
Let $\Et = e_1, e_2, \dots, e_i, \dots, e_{2n-2}$. The prefix $e_1, e_2, \dots, e_{i}$ of $\Et$ up to the 
$i$-th element is denoted by $\Et(i)$. For the proof, we introduce a few auxiliary lemmas.

\begin{lemma}[Claim 3.3 in \cite{DP21}] \label{lma:oddevenDP21}
Let $F$ be any induced cutset of $G$, and let $n_v(F)$ be the number of the edges in $F$ 
lying on the path from the root to $v$ in $T$. Let $V_0 = \{v \in V_G \mid \text{$n_v(F)$ is even} \}$ and
$V_1 = \{v \in V_G \mid \text{$n_v(F)$ is odd} \}$. Then $(V_0, V_1)$ is the cut of $G$ induced by
$F$.
\end{lemma}

\begin{lemma} \label{lma:oddeven}
Assume that $S \subseteq V_T$ contains the root $r$ of $T$.
For any vertex $v \in S$, $|\Et(c(v)) \cap \Cutset_{\vec{T}}(S)|$ is even if $v \in S$, or odd otherwise.
\end{lemma}

\begin{proof}
We first show that any vertex $v$ satisfies that $n_v(\Cutset_{T}(S))$ is even if and only if $v \in S$.
Let $F' \subseteq E_G$ be the cutset associated with the cut $(S, V_G \setminus S)$.
By Lemma~\ref{lma:oddevenDP21}, for any vertex $v \in S$, $n_v(F')$ has the same parity. Since $n_r(F')$ is
obviously even, we can conclude that $n_v(F')$ is even for any $v \in S$. It also implies 
that $n_v(F')$ has the odd parity for any $v \in V_G \setminus S$. By definition, $n_v(F') = n_v(\Cutset_T(S))$
obviously holds. 

Next, we prove the statement of the lemma. For any undirected edge $e \in E_{T}$, 
let $e^+$ and $e^-$ be the downward and upward directed edges
in $E_{\vec{T}}$ corresponding to $e$ respectively. Since $e^+$ always precedes $e^-$ in $\Et$, only the following 
three cases can occur: (1) both $e^+$ and $e^-$ precedes $e_{c(v)}$ in $\Et$, (2) both $e^+$ nor $e^-$ follows 
$e_{c(v)}$ in $\Et$, or (3) $e^+$ precedes or is equal to $e_{c(v)}$ in $\Et$ and $e^-$ follows. 
Since the edge $e \in \Cutset_{T}(S)$ to which the case (1) or (2) applies does not affect the parity of $|\Et(c(v)) \cap \Cutset_{\vec{T}}(S)|$, the parity of $|\Et(c(v)) \cap \Cutset_{\vec{T}}(S)|$ is equal 
to the number of the edges $e \in \Cutset_{T}(S)$ to which the third case applies.
Since the case (3) applies to an edge $e$ if and only if $e$ is on the path from $r$ to $v$, 
we obtain $|\Et(c(v)) \cap \Cutset_{\vec{T}}(S)| = n_v(\Cutset_{T}(S))$. That is, it has the even parity. The lemma is proved.
\end{proof}

Now we are ready to prove Lemma~\ref{lma:cutspace}.
\begin{rlemma}{lma:cutspace}
For any vertex subset $S \subseteq V_G$ and edge subset $E' \subseteq E_{\Gext}$, the following equality holds.
\begin{align*}
\Cutset_{E'}(S) = E' \cap \left(\symdiff_{e \in \Cutset_{\vec{T}}(S), z \in \{x, y\}} \  {\Hp(z, c(e))} \right).
\end{align*}
\end{rlemma}

\begin{proof}
Let $\Qcal_x = \{\Hp(x, c(e)) \mid e \in \Cutset_{\vec{T}}(S)\}$, $\Qcal_y = \{\Hp(y, c(e)) \mid e \in \Cutset_{\vec{T}}(S)\}$, and $Q = \symdiff_{Q' \in \Qcal_x \cup \Qcal_y}  Q'$.
For any $(u, v) \in \Cutset_{E'}(S)$, exactly one of $u$ and $v$ belongs to $S$ and the other one belongs to $V_G \setminus S$.
By symmetry, we assume $u \in S$ and $c(u) < c(v)$ wlog. Lemma~\ref{lma:oddeven} implies $|\Et(c(u)) \cap \Cutset_{\vec{T}}(S)|$ is even, and $|\Et(c(v)) \cap \Cutset_{\vec{T}}(S)|$ is odd. It implies that $(u, v)$ lies in the two regions respectively defined as the intersection of an even number of halfspaces in $\Qcal_y$
and as the intersection of an odd number of halfspaces in $\Qcal_x$. Since $\Qcal_x$ and $\Qcal_y$ are 
disjoint, $(u, v)$ lies in the region defined as the intersection of an odd number of halfspaces in 
$\Qcal_x \cup \Qcal_y$, i.e., it is contained in $Q$.
Similarly, if $(u, v)$ is not an edge in $\Cutset_{E'}(S)$, the parities of 
$|\Et(c(u)) \cap \Cutset_{\vec{T}}(S)|$ and $|\Et(c(v)) \cap \Cutset_{\vec{T}}(S)|$ becomes the same, and thus
$e$ lies in the region defined as the intersection of an even number of halfspaces in $\Qcal_x \cup \Qcal_y$, 
and thus not contained in $Q$. The lemma is proved.
\end{proof}

Next, we focus on the deterministic $\epsilon$-net construction. 
We first quote a known deterministic construction for axis-aligned rectangles.

\begin{lemma}[Mustafa, Dutta, and Ghosh~\cite{MDG18}] \label{lma:optimalnet}
Let $\epsilon > 0$. There exists a deterministic polynomial-time algorithm of constructing an $\epsilon$-net of size 
$O(\log\log N / \epsilon)$ for any $N$ point set and all axis-aligned rectangles. 
\end{lemma}

As we mentioned in Section~\ref{subsec:sparsification}, the polynomial of the construction time has a high exponent, and thus we consider a slightly weaker but much faster solution.
Let $P$ be any set of points in a 2D-range $R = [a_1, a_2] \times [b_1, b_2]$. 
For simplicity, we assume that $|P|$ is a power of two in the lemma below, 
but it is not essential.

\begin{lemma} \label{lma:threeside}
Let $P$ be any point set in $R = [a_1, a_2] \times [b_1, b_2]$, $0 < \epsilon < 1$, 
and, $M \in [a_1, a_2]$. There exists a $O(|P|\log |P|)$-time 
algorithm of computing a subset $P^{\ast} \subseteq$ of size at most $\epsilon|P|$ 
such that any axis-aligned rectangle $X$ crossing the vertical line
$x = M$ and satisfying $|X \cap P| \geq 6/\epsilon$ necessarily contains at least one point in $P^{\ast}$.
\end{lemma}

\begin{proof}
The construction follows the technique by Kulkarni and Govindarajan~\cite{KG10}. For simplicity, we assume that 
any two points in $P$ have different $y$-coordinates, $2/\epsilon$ is an integer, and $|P|$ is divisible by $2/\epsilon$.
Let $Y_P$ be the sequence of points in $P$ sorted by their $y$-coordinates. We split $Y_P$ into 
$\epsilon|P|/2$ subsequences $Y^1_P, Y^2_P, \dots$ of the length $2/\epsilon$. 
For each $Y^i_P$, we define 
$p^-_i \in Y^i_P$ as the point with the maximum $x$-coordinate not exceeding $M$, and $p^+_i \in Y^i_P$ as the one 
with the minimum $x$-coordinate not lower than or equal to $M$. We construct $P^{\ast}$ as the union of $\{y^+_i, y^-_i\}$ 
for all $i$. 

It is easy to check that the constructed point set $P^{\ast}$ satisfies the condition of 
the lemma: The size of $P^{\ast}$ is obviously bounded by $\epsilon|P|$.
We define the $y$-range $[b^i_1, b^i_2]$ of $Y^i_P$ as the minimal interval containing the $y$-coordinates 
of all the points in $Y^i_P$. Consider any rectangle
$X = [a_1, a_2] \times [b_1, b_2]$ such that $|X \cap P| \geq 6/\epsilon$. Then there exists at least one 
subsequence $Y^i_p$ such that at least one point in $Y^i_P$ is contained in $X \cap P$ and the 
$y$-range of $Y^i_P$ is covered by $[b_1, b_2]$. For such an $i$, either $p^-_i$ or $p^+_i$ must be contained in $X$.
\end{proof}

We obtain the deterministic algorithm of constructing a $O(\log |P|/ |P|)$-net of a constant 
fraction size for any point set $P$ and axis-aligned rectangles 
in near linear time. 

\begin{lemma} \label{lma:rectangle}
Let $P$ be any set of points in $[a_1, a_2] \times [b_1, b_2]$, and $N$ be any upper bound of $|P|$. 
There exists a deterministic algorithm $\NetFind(N, P)$ which constructs a 
$(12\log N / |P|)$-net of size at most 
$|P|\log |P|/(2\log N)$ for $P$ and all axis-aligned rectangles in $O(|P|\log |P| \log N)$ time.
\end{lemma}

\begin{proof}
We first present the algorithm. It is based on the divide and conquer approach as follows:
\begin{enumerate}
\item If $|P| \geq 12 \log N$: find the vertical line $x = M$ bisecting 
$P$ into two equal-size subsets $P_0$ and $P_1$ (with an arbitrary tie-breaking rule for the points on $x = M$). 
Let $R_0 = [a_1, M] \times [b_1, b_2]$ and $R_1 = [M, a_2] 
\times [b_1, b_2]$. Then $\NetFind(N, P)$ outputs the union of the following four subsets of $P$:
\begin{itemize}
\item The outputs of $\NetFind(N, P_0)$ and $\NetFind(N, P_1)$.
\item The point set $P^{\ast}$ obtained from $P$ by Lemma~\ref{lma:threeside} for $\epsilon = 1/(2\log N)$ and 
$x = M$.
\end{itemize}
\item If $|P| < 12 \log N$: output the empty set. 
\end{enumerate}
Let $P'$ be the output in the run of $\NetFind(N, P)$.
The proof is based on the induction on the size of $P$.

\noindent
(Basis) $|P| < 12 \log N$ holds: Then the case 2 of the algorithm applies. The output size
is trivially bounded by $|P|\log |P|/(2\log N) \geq 0$. Since $|P| < 12 \log N$ holds, there is no
axis-aligned rectangle containing more than or equal to $12\log N$ points. Hence the constructed output (i.e.,
the empty set) is a $(12\log N/ |P|)$-net. 

\noindent
(Inductive Step): Let $P$ be the set of points 
such that $|P| \geq 12 \log N$ holds. Consider any axis-aligned rectangle $X$ such that $|X \cap P| \geq 12\log N$ 
holds.  We show that $X$ necessarily contains a point in $P'$.
By the induction hypothesis, both $\NetFind(N, P_0)$ and 
$\NetFind(N, P_1)$ correctly computes a $(12\log N / |P_0|)$-net and a $(12\log N / |P_1|)$-net. Hence 
if $X$ is contained either $R_0$ or $R_1$, $X$ necessarily contains a point in $P'$. If $X$ 
intersects $x = M$, it also intersects $P'$ by Lemma~\ref{lma:threeside}. Hence the constructed $P'$ is 
a $(12\log N / |P_1|)$-net. We bound the output size $|P'|$. By the induction hypothsis,
the output sizes of $\NetFind(N, P_0)$ and $\NetFind(N, P_1)$ are respectively bounded by 
$|P| \log (|P|/2) /(4\log N)$. The size of $P^{\ast}$ is bounded by $|P|/(2\log N)$. Summing up them,
we obtain
\begin{align*}
    |P'| &\leq 2 \cdot \frac{|P|\log\frac{|P|}{2}}{4\log N} + \cdot \frac{|P|}{2 \log N} \\
         &\leq \frac{|P|(\log |P| - 1)}{2\log N} + \frac{|P|}{2 \log N} \\
         &\leq \frac{|P|\log |P|}{2\log N}.
\end{align*}
It is easy to bound the running time because the total running time of all recursive calls at the same depth 
is bounded by $O(|P| \log N)$. The lemma is proved.
\end{proof}

Lemma~\ref{lma:e-net} is obviously obtained from Lemma~\ref{lma:optimalnet} and Lemma~\ref{lma:rectangle}.
We finally show the main lemma of this section.
\begin{rlemma}{lma:goodhierarchy}
There exists two deterministic algorithms respectively constructing a $(\Scal_{f, T}, k)$-good hierarchy with the following performance guarantees:
\begin{itemize}
\item $k = O(f^2 \log n)$ and the construction time is $\tilde{O}(m)$.
\item $k = O(f^2 \log \log n)$ and the construction time is $\Poly(m)$.
\end{itemize}
\end{rlemma}

\begin{proof}
We only show that first construction, but the second construction is proved in the same way.
Consider the construction of $E_{i+1}$ from $E_i$. The algorithm first maps all the edges in $E_i$ into 
the space $[1, 2n - 2]^2$. Applying the algorithm of Lemma~\ref{lma:rectangle} to $E_i$ with $P = E_i$ and 
$N = |E_i|$, we obtain an $(6\log |E_i| / |E_i|)$-net $E_{i+1}$ of a constant fraction size for $E_i$ and 
all axis-aligned rectangles. As mentioned in Section~\ref{subsec:sparsification}, $E_{i+1}$ works as a 
$(6(2f + 1)^2 \log n/|E_i|)$-net for $\Hcal_{2f}$ and thus it satisfies the condition of 
$(\Scal_{f, T}, 6(2f + 1)^2 \log n)$-good edge hierarchy. Since $E_{i+1}$ is a subset
of $E_i$ with a constant fraction size, the depth $h$ of hierarchy is bounded by $O(\log n)$. The construction time
of $E_{i+1}$ from $E_i$ follows the running time of the algorithm of Lemma~\ref{lma:rectangle}, i.e., 
it takes $\tilde{O}(m)$ time. Hence the total running time is $\tilde{O}(m)$.
\end{proof}

\subsection{Fast Query Processing} \label{subsec:adaptive}

We construct a refined query processing algorithm for our tree edge $f$-FTC labeling scheme. 
Let $\Gext = G - E_T$ for short. Throughout
this section, we assume that the tree edge $f$-FTC labeling scheme is implemented with 
any adaptive $\Scal_{f, T}$-outdetect labeling scheme $(\Olabel_{\Gext}, \Odec)$ of $\tilde{O}(f^b)$-bit
label size which admits $\tilde{O}(|\Cutset_T(S)|^c)$ decoding time for a given query 
$S \subseteq V_{\Gext}$, as explained in Section~\ref{sec:improvement}. Similarly as the original one, 
the refined 
algorithm also iteratively merges the vertices of the component graph $T / \Ccal(F)$. 
It manages a collection $\Xcal$ of disjoint subsets of $\Ccal(F)$ throughout the procedure.  
A subset $S \subseteq \Ccal(F)$ in $\Xcal$ is called a \emph{component fragment} (note that $S$
is not a subset of $V_G$). For any component fragment $S \subseteq 
\Ccal(F)$, we define $V(S) \subseteq V_G$ as $V(S) = \bigcup_{C \in S} C$.
As a loop invariant, the algorithm guarantees that $V(S)$ for any component 
fragment $S \in \Xcal$ induces a connected subgraph of $G - F$. 
Each component fragment $S \subseteq \Ccal(F)$ is stored as 
the triple $(S, \Cutset_T(V(S)), \Olabel_{\Gext}(V(S)))$ in $\Xcal$. We denote this triple 
associated with $S \subseteq \Ccal(F)$ by $\tau(S)$. The whole structure of the algorithm is stated
below:
\begin{enumerate}
    \item  Initially, we set $\mathcal{X} = \{\tau(\{C\}) \mid C \in \Ccal(F)\}$. 
    \item  In each iteration, we pick up $\tau(S)$ such that $|\Cutset_T(V(S))|$ is the smallest, 
    and find an outgoing edge of $S$ by decoding $\Olabel_{\Gext}(V(S))$, which is
    obtained from $\Olabel_{\Gext, K}(V(S))$ stored in $\tau(S)$.
    \item If no outgoing edge of $S$ is found, we remove $\tau(S)$ from $\Xcal$ and go to the next iteration. 
    Otherwise, let $S'$ be the component fragment the outgoing edge from $S$ reaches.
    \item If $S$ and $S'$ contains $s$ and $t$ respectively, the procedure terminates with returning 
    true. Otherwise, the algorithm deletes $\tau(S)$ and $\tau(S')$ from $\Xcal$, and newly insert 
    the entry of $\tau(S'')$ for $S'' = S \cup S'$. The entry $\tau(S'')$ is computed as $\tau(S'') = 
    (S'', (\Cutset_T(V(S)) \cup \Cutset_T(V(S'))) \setminus (\Cutset_T(V(S)) \cap \Cutset_T(V(S'))), 
    \Olabel_{\Gext}(V(S)) + \Olabel_{\Gext}(V(S')))$. After the insertion, the algorithm 
    proceeds to the next iteration unless $|\Xcal| = 1$ holds. If $|\Xcal| = 1$ holds, the algorithm 
    terminates with returning false (this case occurs only when the component fragment containing $s$ or 
    $t$ is discarded).
\end{enumerate}
To implement the algorithm above efficiently, we manage $\Xcal$ by the heap which 
supports $O(\log |\Xcal|)$-time insert, delete, and search of the element having the 
minimum cutset. Each cutset associated with an element in $\Xcal$ is stored as the bit vector of length 
$|F|$ and the additional integer value representing the size of the stored cutset. 
This data structure obviously supports union and intersection in $O(|F|)$ time, as well as getting 
the cutset size in $O(1)$ time. Each fragment $S$ associated with a triple in $\Xcal$ is managed by 
any disjoint-set data structure (e.g., union-find) over $\Ccal(F)$. Combining this structure with Proposition~\ref{prop:identifycomponenttree}, one can determine the fragment $S \in \Xcal$ containing 
a given vertex $u \in V_G$ from $\Alabel_T(u)$ in $O(\log |F|)$ time (i.e., 
identify $C \in \Ccal(F)$ containing $u$ first, and then identify $S \in \Xcal$ containing $C$). 
With support of all the data structures above, we can implement one iteration of 
the refined procedure in $\tilde{O}(|F|^b + |\Cutset_T(V(S))|^c)$ time.
The initialization of $\Xcal$ is implemented in 
$\tilde{O}(|F|^{b + 1})$ time. We show that the refined algorithm runs in $\tilde{O}(|F|^c)$ time 
in total.

\begin{rlemma}{lma:improveddecoding}
Assume that there exists a $\Scal_{f, T}$-outdetect labeling scheme of label size $\alpha = \tilde{O}(f^b)$ and 
decoding time $\beta = \tilde{O}(f^c)$. Then there exists a $f$-FTC labeling scheme of $O(\alpha + \log n)$-bit label size 
and $\tilde{O}(|F|^{b+1} + |F|^c)$ decoding time. The resultant $f$-FTC labeling scheme is deterministic if the corresponding
$\Scal_{f,T}$-outdetect labeling scheme is deterministic.
\end{rlemma}

\begin{proof}
Consider the refined query processing algorithm above.
We denote by $\Xcal_i$ the set $\Xcal$ at the beginning of the $i$-th iteration, and 
let $\Ycal_i = \{S \mid (S, \cdot, \cdot) \in \Xcal_i \}$. Since $\Ycal_i$ is a 
disjoint collection of component fragments and each $S \in \Ycal_i$ satisfies $\Cutset_T(V(S)) 
\subseteq F$, we have $\sum_{S \in \Ycal_i} |\Cutset_{T}(S)| \leq 2|F|$. 
Let $S_1, S_2 \dots, S_x$ be the component fragments chosen in each iteration ($x \leq |F|$). 
Since the algorithm chooses $S_i$ minimizing 
$|\Cutset_{T}(V(S_i))|$, we have $|\Cutset_{T}(V(S_i))| \leq 2|F| / |\Ycal_i|$. As discussed in this section,
the detection of an outgoing edge of $S_i$ takes $\tilde{O}(|\Cutset_{T}(V(S_i))|^c)$ time. 
Since at least one component in $\Ycal_i$ is merged or discarded in the $i$-th iteration, 
we have $|\Ycal_{i+1}| \leq |\Ycal_{i}| - 1$. 
The computation time excluding that for the outgoing edge detection is $O(|F|^b\log n)$ per one 
iteration. The total running time is bounded as follows: 
\begin{align*}
    \sum_{1 \leq i \leq |F|} \tilde{O}\left(\left(\frac{|F|}{i}\right)^c + |F|^b\right) \\
   & \hspace{-30mm} \leq  
    \tilde{O}(|F|^c) \cdot \sum_{1 \leq i \leq |F|}
    \frac{1}{i^c}  + \tilde{O}(|F|^{b + 1})\\
    & \hspace{-30mm} \leq \tilde{O}(|F|^c + |F|^{b+1}).
\end{align*}
\end{proof}

\section{Distributed Construction}
\label{sec:distconst}

In this section, we explain how our deterministic $f$-FCT labeling scheme is constructed in
the standard CONGEST model, i.e., the round-based synchronous system with the 
$O(\log n)$-bit message size bound. For the input graph $G$,
we fix $T$ as its BFS tree for an arbitrary chosen root. Since the corresponding auxiliary 
graph $G'$ is easily simulated on the top of the original graph, the behavior of the 
proposed algorithm is described as the message passing on $G'$. The spanning tree of $G'$
transformed from $T$ is denoted by $T'$. 

\paragraph{Construction of Ancestry Labels}
The construction by Kannan, Naor, and Rudich is to assign each node and edge
with the pair of its pre-order and post-order in the Euler-tour traversal of $T'$ starting 
from the root. For any two labels $(a, b)$ and $(c, d)$ assigned to $u$ and $v$, $u$ is an
ancestor of $v$ if and only if the interval $[a, b]$ contains $[c, d]$. In the following argument,
we focus on the computation of the pre-orders and post-orders of edges in $T'$. The computation 
of vertex orders is processed similarly. 
%We first assign each edge with the pre-order and post-order in the 
%Euler-tour traversal of $T'$. 
For every edge $e$, the algorithm computes the number of edges in the subtree 
$T'(e)$, which is implemented by the subtree-sum aggregation over $T'$, taking $O(D)$ rounds. The 
twice of the computed value, denoted by $\mathsf{gap}(e)$, is equal to the gap between the pre-order and the post-order of $e$.
Then the algorithm determines the pre-orders and post-orders of all tree edges from the root side.
Assume that we have fixed the orders $(a, b)$ of an edge $e$, and let $e_1, e_2, \dots, e_j$ be the set 
of children edges of $e$. Then the pre-order of $e_1$ is obviously $a+1$, and its post-order is 
$a + 2 + \mathsf{gap}(e_1)$. The orders of $e_2, e_3, e_4, \dots, e_j$ are decided similarly.

\paragraph{Construction of Outdetect Labels}
Given a $(\Scal_{f, T'}, O(f^2 \log n))$-good hierarchy, it is trivial to compute 
$k$-threshold outdetect labels. The label $g(e)$ assigned to each non-tree edge is computed locally,
and the label to each vertex is computed by the subtree-sum aggregation of message size 
$\tilde{O}(f^2)$. Using the standard pipeline technique, it is processed in 
$\tilde{O}(D + f^2)$ rounds.

\paragraph{Distributed Construction of a $(\Scal_{f, T'}, O(f^2 \log n))$-good hierarchy}
To construct such a good hierarchy, it suffices to implement $\NetFind$ in the CONGEST model.
As a preprocessing, the algorithm first computes the $x$-$y$ coordinates of all non-tree edges, which 
is realized in the same way as the construction of the ancestry labels. After this preprocessing,
one can assume that both of the endpoints of any non-tree edge $e$ knows the coordinate of $e$.
For simplicity of the argument, we assume that $m' = |E_{G'}| - |E_{T'}|$ is a power of $2$.
Our implementation processes the calls of $\NetFind$ 
at the same recursion level in parallel. It is easy to check that the following two conditions are satisfied:
\begin{itemize}
\item For any call of $\NetFind(N, P)$ at the recursion level $j > (\log m')/2$, we have $|P| \leq \sqrt{m'}$.
In addition, the information of the points (edges) in $P$ is stored at nodes on a consecutive subsequence 
(denoted by $\mathsf{seq}(P)$) of the Euler-tour traversal of $T'$.
\item For any two calls of $\NetFind(N, P_1)$ and $\NetFind(N, P_2)$ at the same recursion level, 
$\mathsf{seq}(P_1)$ and $\mathsf{seq}(P_2)$ are edge disjoint under the treatment of two edges $(u, v)$ and $(v, u)$ as 
distinct ones.
\end{itemize}
Since $\mathsf{seq}(P)$ induces a connected subtree of $T'$, for the invocation of $\NetFind(N, P)$,
every node in $\mathsf{seq}(P)$ can aggregates whole information of $P$ in $O(\sqrt{m'} + D)$ rounds. Hence each node
can execute the centralized version of $\NetFind$ locally. Due to the second condition above, the aggregation tasks for
all $P$ at the same recursion level is efficiently processed in parallel: each edge $e \in E_{T'}$ is contained at most two induced
subtrees, and thus the total running time is still bounded by $O(\sqrt{m'} + D)$. Consequently, all the 
invocations of $\NetFind(N, P)$ at the same recursion level $j > (\log m')/2$ are processed in $O(\sqrt{m'} + D)$ rounds.
For the recursion level $j < (\log m')/2$. The total number of invocations is $O(\sqrt{m'})$, and thus the algorithm
sequentially processes each invocation. The main body of $\NetFind(N, P)$ is the construction of the
point set $P'$ shown in Lemma~\ref{lma:threeside}. To identify the set $Y^1_P, Y^2_P, \dots, $, it suffices to compute the 
order of each non-tree edge in the sequence $Y$, which is computed in $O(D)$ rounds 
similarly with the construction of ancestry labels. Next, the algorithm must decide $p^-_i$ and $p^+_i$ for each $Y^i_P$, which 
is implemented by the information exchange within the subgraph induced by $\mathsf{seq}(Y^i_P)$. This part also needs only 
$O(D + \epsilon^{-1})$ rounds. We use this construction with $\epsilon = O(1/\log m')$, 
and thus the running time of $\NetFind(N, P)$ except for the recursive calls is $\tilde{O}(D)$. In total, the running time of $\NetFind(N, P)$ for 
$P = E_{G'} \setminus E_{T'}$ is $\tilde{O}(\sqrt{m'}D)$. To construct whole $(\Scal_{f, T'}, c f^2 \log n)$-good hierarchy, $O(\log n)$ repetition 
of invoking $\NetFind$ suffices. Consequently, we obtain the following lemma
\begin{lemma} \label{lma:disthierarchy}
Let $T'$ be any BFS tree of the auxiliary graph $G'$ of the input graph $G$.  
There exists a deterministic CONGEST algorithm of constructing a $(\Scal_{f, T’}, O(f^2 \log n)$-
good hierarchy in $\tilde{O}(\sqrt{m}D)$ rounds.
\end{lemma}

This lemma obviously deduces the theorem below:
\begin{theorem} \label{thm:distcost}
There exists a deterministic CONGEST algorithm of constructing $f$-FTC labels for all vertices and edges in $\tilde{O}(\sqrt{m}D + f^2)$ rounds. 
\end{theorem}

\section{Concluding Remarks}
\label{sec:conclusion}

This paper presented a new deterministic $f$-FTC labeling scheme which attains $O(f^2 \Polylog(n))$-bit label 
size, polynomial-time construction, and $\tilde{O}(\Poly(|F|))$-time query processing time for 
a given faulty edge set $F$. This is the first deterministic and polynomial-time $f$-FTC labeling scheme 
with a non-trivial label size. The scheme is developed on the top of a general framework, and only by 
the modification of graph sparsification, we can also obtain a randomized $f$-FTC labeling scheme which is 
competitive to the original Dory-Parter scheme and 
attains an adaptive query processing time. The key technical ingredient is a new deterministic
$\Scal$-outdetect labeling scheme based on error-correcting codes. 
From the authors' perspective, our results pose a few promising future research directions. We conclude 
this paper with summarizing them.

\begin{itemize}
    \item Is it possible to develop a deterministic algorithm yielding better edge hierarchies, i.e.,
    the hierarchy such that for any $S$ there exists $i$ satisfying $0 < |\Cutset_{E_{i}}(S)| = o(f^2 \log n)$?
    Our framework automatically deduces a deterministic $f$-FTC labeling
    scheme with an improved label size if such an algorithm is found. 
    \item With respect to the construction time in the CONGEST model, our deterministic scheme
    still has a large gap with the known randomized construction, only taking  
    $\tilde{O}(f + D)$ rounds. Is it possible to obtain the deterministic $f$-FTC labeling scheme of $O(\Poly(f, \log n))$-bit label size 
    which is implemented in the CONGEST model with $\tilde{O}(\Poly(f) + D)$ or $\tilde{O}(\Poly(f) \cdot D)$ rounds?
    \item Can we obtain any non-trivial lower bound for the label size of $f$-FTC labeling schemes with full query support? 
    It seems plausible that the $\Omega(f)$-bit lower bound holds, but no promising way of proving this is found so far.
    \item Can we use our technique to obtain a (deterministic) fault-tolerant connectivity labeling scheme for 
    \emph{vertex faults}? As pointed out in~\cite{PP22}, there exists a large technical gap between edge fault tolerance and vertex fault tolerance. It is still open to obtain a scheme with a label size sublinear of $n$, even for moderately small $f$ (e.g., $f = O(\log n))$.
    \item Can our technique be exported to other applications of $\Scal$-outdetect labeling schemes,
    such as centralized fault-tolerant connectivity oracles~\cite{DP20}, distributed computation of sparse spanning subgraphs~\cite{KKT15,GP18,GK18,MK21} or small cut detection~\cite{PT11,PP22-2},
    and dynamic algorithms~\cite{KKM13,GKKT15}, for obtaining any improved result? 
\end{itemize}

\newcommand{\etalchar}[1]{$^{#1}$}

\appendix

\section{Randomized Construction of Edge Set Hierarchy}
\label{appendix:randomized}
Comparing the quality of labeling schemes with the label size, decoding time, and construction time,
our deterministic construction is competitive but certainly worse than the known randomized scheme. 
However, most of high costs incurred by our construction is derived from the construction of 
edge set hierarchies. As mentioned in Section~\ref{subsec:sparsification}, a simple edge sub-sampling strategy
suffices to construct good hierarchies.
\begin{proposition} \label{prop:randomized}
Let $E_{\Gext} = E_0 \supseteq E_1 \supseteq E_2 \supseteq, \dots, \supseteq E_h = \emptyset$ be 
the edge set hierarchy such that $E_{i+1}$ is obtained by sampling each
edge in $E_i$ independently with probability $1/2$. Then with probability $1 - 1/n^{O(1)}$,
this hierarchy is $(\Scal_{f, T}, 5f\log n)$-good.
\end{proposition}

\begin{proof}
Consider the construction of $E_{i+1}$ from $E_i$. If $|E_i| \leq 5 f\log n$, we decide $E_{i+1} = \emptyset$.
Otherwise, sample each edge in $E_i$ with probability $1/2$. We show that
with high probability, $|\Cutset_{E_{i+1}}(S)| > 0$ holds for any $S \in \Scal_{f, T}$ satisfying 
$|\Cutset_{E_{i}}(S)| > 5f\log n$. The probability that no edge in 
$\Cutset_{E_i}(S)$ is added to $E_{i+1}$ is at most $(1/2)^{5f\log n} = 1/n^{5f}$.  
Since the cardinality of $\Scal_{f, T}$ is bounded by $\binom{|E_T|}{f} = O(n^{f})$, by the union bound 
argument, we conclude that $|\Cutset_{E_{i+1}}(S)| > 0$ holds for any $S \in \Scal_{f, T}$ 
satisfying $|\Cutset_{E_i}(S)| > 5f \log n$ with probability $1 - O(1/n^{5})$. This 
statement holds for all possible $c$ with probability at least $1 - O(1/n^{5})$. That is, 
the probability that $E_{i+1}$ does not satisfies the second condition of 
Definition~\ref{def:goodness} is $O(1/n^4)$. The first condition is satisfied with high probability because
one can show $|E_{i+1}| \leq 3|E_i|/4$ with probability $1 - O(1/n^2)$ by the straightforward 
application of Chernoff bound. Applying the union bound again on the failing events of 
the first and second conditions for all $i$, we obtain the proposition.
\end{proof}

\section{Adaptive Decoding of Deterministic $\Scal_{f, T}$-outdetect labeling scheme based on the Reed-Solomon Code}
\label{appendix:adaptiveRS}
In this appendix, we show that our $\Scal_{f, T}$-outdetect labeling scheme attains the adaptiveness
without any modification or transformation. The key idea is a nice property 
of Reed-Solomon Code: We define $\RSlabel{k}_{H, k'}$ for $k' \leq k$ as 
the labeling function which assign $v \in V_H$ with the prefix of $\RSlabel{k}_{H}(v)$ up to 
the $k'$-th element. Then the following proposition holds:
\begin{proposition} \label{prop:rsprefix}
For any $k' \leq k$, $\RSlabel{k}_{H, k'} = \RSlabel{k'}_{H}$ holds.
\end{proposition}
\begin{proof}
The proof trivially follows the definition of the parity check matrix 
$C_{2k}$. Let $c_{ij}$ be the $(i, j)$-element of $C_{2k}$ 
($0 \leq i \leq |\Ecal|, 0 \leq j \leq 2k -1$) and $\omega$ be a
primitive element of $\Field$. Then the each element of $C_{2k}$ is defined as 
$c_{ij} = \omega^{ij}$. That is, the submatrix formed by the first $2k'$ columns of 
$C_{2k}$ is is equal to $C_{2k'}$.
\end{proof}
Intuitively, the proposition above implies that the $O(k' \log n)$-bit prefixes of the labels assigned 
by our $k$-threshold outdetect labeling scheme also work as the labels of the $k'$-threshold 
outdetect labeling scheme. It is also possible to make the sparsification hierarchy adaptive 
because our construction does not use the upper bound $f$ at all, i.e., the construction is universal 
for every $f$: We explained in the previous section that $E_{i+1}$ is constructed in the way that 
it becomes the hitting set of $\Zcal_{i, f, c f^2 \log n}$ (where $c$ 
is a hidden constant). In reality, it also becomes the hitting set of $\Zcal_{i, f', c (f')^2 }$ 
for any $f' > 0$. Hence it is guaranteed that for any $S \in \Scal_{|\Cutset_T(S)|, T}$ there exists an index 
$i$ such that $0 < |\Cutset_{G_i}(S)| \leq c |\Cutset_T(S)|^2 \log n$ holds (in the deterministic case).

%%%%%%%%%%%%%%%%%%%%%%%%%%% END %%%%%%%%%%%%%%%%%%%%%%%%%%%%%

\end{document}